\documentclass[aps,pra,reprint,superscriptaddress,nofootinbib,onecolumn]{revtex4}

\usepackage{amsmath}
\usepackage{amsfonts}
\usepackage{amsthm}
\usepackage{amssymb}
\usepackage{graphicx}
\usepackage{enumerate}
\usepackage{color}
\usepackage[T1]{fontenc}
\usepackage{mathptmx}
\usepackage{braket}
\usepackage{todonotes}
\usepackage{soul}
\usepackage{subfigure}
\usepackage{mathrsfs}
\usepackage{eucal}
\usepackage{subfigure}
\usepackage{float}
\usepackage{bbold}
\usepackage{caption}
\usepackage{hyperref}
\usepackage{lineno}
\graphicspath{{figures/}}

\DeclareMathOperator{\Tr}{Tr}

\newtheorem{theorem}{Theorem}
\newtheorem*{theorem*}{Theorem}
\newtheorem{definition}{Definition}
\newtheorem*{definition*}{Definition}
\newtheorem{lemma}{Lemma}
\newtheorem*{lemma*}{Lemma}
\newtheorem{corollary}{Corollary}

\newtheorem{observation}{Observation}
\newtheorem{construction}{Construction}
\newtheorem*{construction*}{Construction}

\newcommand{\revise}[1]{\textcolor{black}{#1}}

\def\zo{$\{0,1\}$}
\def\>{\rangle}
\def\<{\langle}

\begin{document}

\title{Optimal Measurement Structures for Contextuality Applications}
\author{Yuan Liu}
\affiliation{Department of Computer Science, The University of Hong Kong, Pokfulam Road, Hong Kong}
\author{Ravishankar Ramanathan}
\email{ ravi@cs.hku.hk}
%\thanks{ravi@cs.hku.hk}
\affiliation{Department of Computer Science, The University of Hong Kong, Pokfulam Road, Hong Kong}

\author{Karol Horodecki}
\affiliation{Institute of Informatics, National Quantum Information Centre, Faculty of Mathematics, Physics and Informatics, University of Gda\'{n}sk, Wita Stwosza 57, 80-308 Gda\'{n}sk, Poland}
\affiliation{International Centre for Theory of Quantum Technologies, University of Gda\'{n}sk, Wita Stwosza 63, 80-308 Gda\'{n}sk, Poland}
\author{Monika Rosicka}
\affiliation{Institute of Informatics, National Quantum Information Centre, Faculty of Mathematics, Physics and Informatics, University of Gda\'{n}sk, Wita Stwosza 57, 80-308 Gda\'{n}sk, Poland}
\author{Pawe{\l} Horodecki}
\affiliation{International Centre for Theory of Quantum Technologies, University of Gda\'{n}sk, Wita Stwosza 63, 80-308 Gda\'{n}sk, Poland}
\affiliation{Faculty of Applied Physics and Mathematics, Gda\'{n}sk University of Technology, Gabriela Narutowicza 11/12, 80-233 Gda\'{n}sk, Poland}

%\affiliation{xxx}
%%%%%%%%%%%%%%%%%%%%%%%%%%%%%%%

\begin{abstract}
% Applications of the foundational Kochen-Specker (KS) theorem  have attracted much interest recently. Here, we show that measurement structures within KS proofs termed gadgets provide an optimal toolbox for contextuality applications including (i) constructing classical channels exhibiting entanglement-assisted advantage in zero-error communication, (ii) identifying large separations between quantum theory and binary generalised probabilistic theories, and (iii) finding optimal tests for contextuality-based semi-device-independent randomness generation. Furthermore, we introduce a higher-order generalisation of gadgets that we show exist within general KS proofs, and use them to construct novel KS proofs as well as general
% state-independent contextual (SIC) proofs.  

%pinpoint the role of gadgets in achieving the essential KS contradiction and introduce and study a higher-order generalization. 
%The constructions developed here can help resolve the remaining open problems regarding minimal KS proofs.

The Kochen-Specker (KS) theorem is a corner-stone result in the foundations of quantum mechanics describing the fundamental difference between quantum theory and classical non-contextual theories. Recently specific substructures termed $01$-gadgets were shown to exist within KS proofs that capture the essential contradiction of the theorem. Here, we show these gadgets and their generalizations provide an optimal toolbox for contextuality applications including (i) constructing classical channels exhibiting entanglement-assisted advantage in zero-error communication, (ii) identifying large separations between quantum theory and binary generalised probabilistic theories, and (iii) finding optimal tests for contextuality-based semi-device-independent randomness generation. Furthermore, we introduce and study a generalisation to definite prediction sets for more general logical propositions, that we term higher-order gadgets. We pinpoint the role these higher-order gadgets play in KS proofs by identifying these as induced subgraphs within KS graphs and showing how to construct proofs of state-independent contextuality using higher-order gadgets as building blocks. The constructions developed here may help in solving some of the remaining open problems regarding minimal proofs of the Kochen-Specker theorem.

\end{abstract}
\maketitle
\section{Introduction}
The Kochen-Specker (KS) theorem \cite{KS,Bell} is a cornerstone result in the foundations of quantum mechanics, which delineates the differences between quantum theory and a class of hidden-variable theories obeying the principle of non-contextuality (NCHVTs). NCHVTs assume that outcomes
are pre-assigned to measurements and independent of the
particular contexts in which the measurements are realized.
Informally, the KS theorem states that for every quantum system belonging to a Hilbert space of dimension $d$ greater than two, irrespective of its actual state, a finite set of measurements exists whose results 
%(described by a set of rank-one projectors in quantum mechanics) 
is logically impossible to be assigned of truth value  0 or 1 in a context-independent manner,
% Formally speaking, there are sets of rank-one projectors of Hilbert space dimension $d$ greater than two do not admit a deterministic non-contextual assignment (also known as KS-assignment) which is an assignment of 0 or 1 to each projector 
satisfying (i) Exclusivity: two orthogonal projectors are not allowed to be both assigned 1, and (ii) Completeness: for each $d$ mutually orthogonal projectors, one of them must be assigned 1.

KS sets are not the only means to identify the differences between quantum mechanics and NCHVTs. 
%It was later shown that bounds can be derived on the basis of outcome non-contextuality for certain `non-contextuality (NC) inequalities' that are specific linear combinations of probabilities of events or propositions. A violation of these bounds by quantum probabilities then establishes the distinction between quantum and non-contextual theories. NC inequalities can be of two types: (i) inequalities (such as from KS sets) that are violated by all quantum states of a specific dimension, these are termed statistical state-independent KS arguments, and (ii) inequalities that are only violated by a subset of quantum states, these are termed statistical state-dependent KS arguments. 
An interesting class of statistical state-dependent proofs of contextuality was also presented by Clifton \cite{Clifton93}, Stairs \cite{Stairs}, Hardy \cite{hardy1993nonlocality} and others \cite{Belinfante73, RRHP+20}. In these works, a prediction occurs with certainty in every non-contextual theory (such as the probability of an event being $0$ or $1$), while this is not the case in quantum theory. Such statistical proofs provide a simple and appealing contradiction between quantum and NCHVTs. Considering each projector in a Hilbert space as an atomic proposition, sets of the form $P \rightarrow \overline{Q}$ ($P$ being true implies $Q$ is false) or $P \rightarrow Q$ ($P$ being true implies $Q$ is true) have been termed as gadgets \cite{RRHP+20}, definite-prediction sets \cite{CA96}, bugs or true-implies-false and true-implies-true sets \cite{APSS18}. 
%Considering each vector in a Hilbert space as an atomic proposition, sets that have been considered have been of the form $P \rightarrow \overline{Q}$ ($P$ being true implies $Q$ is false) or $P \rightarrow Q$ ($P$ being true implies $Q$ is true). 
%Note that any Kochen-Specker set itself corresponds to building out of the atomic propositions some compound proposition that is a contradiction. 

As a central result in the foundations of quantum mechanics,
KS contextuality has yielded several exciting applications in quantum information science recently. These applications include entanglement-assisted advantage in zero-error communication \cite{CLMW10}, semi-device-independent randomness generation \cite{U13}, device-independent security \cite{HHHH+10},
universal quantum computation via magic state distillation \cite{Howard}, advantage in communication complexity \cite{gupta2022quantum},
self-testing quantum systems \cite{BRV+19}, etc.

In this paper, we introduce a general class of definite-prediction sets termed \textit{higher-order gadgets} that goes beyond the basic `true-implies-false' and `true-implies-true' structures considered thus far and show how these gadget measurement structures provide an optimal toolbox for a plethora of applications of contextuality \cite{HHHH+10, gupta2022quantum, Howard, CSW14, KC16}. 
% We introduce a general class of definite-prediction sets that goes beyond the basic `true-implies-false' and `true-implies-true' structures considered thus far. 
(i) We show how the entanglement-assisted advantage in zero-error communication, previously discovered for KS proofs alone, persists for the smaller and experimentally feasible classical channels corresponding to gadgets, under a suitable generalisation. (ii) We apply gadgets to provide an experimentally feasible test of a recent result demonstrating that quantum correlations cannot be reproduced by fundamentally binary theories. 
These are a natural class of alternatives to the set of correlations allowed by quantum theory, and are defined as general probabilistic theories that posit that on a fundamental level only measurements with two outcomes exist. 
(iii) We point out that gadget-based contextuality tests allow to certify the maximal amount of $\log d$ bits of randomness from $d$-dimensional systems, making them ideal candidates for contextuality-based semi-device-independent randomness generation. (iv) We also use gadgets to point out a subtle modification to the famous Cabello-Severini-Winter (CSW) graph-theoretic framework of contextuality, namely that the classical value of non-contextuality inequalities does not always equal the weighted independence number of the corresponding orthogonality graph, when the KS rules of Exclusivity and Completeness are enforced.
Furthermore, we show the constructions of definite prediction vector sets corresponding to arbitrary compound proposition, i.e., the entire spectrum of Hardy tests of contextuality from basic `true-implies-false' sets to KS sets. We identify how these vector sets can be found inside general KS proofs, and demonstrate how these can be applied as building blocks for constructing novel KS proofs as well as general state-independent contextual (\revise{SI-C}) proofs. Consequently, since higher-order gadgets form an essential ingredient in the constructions of KS proofs, minimal constructions of gadgets may be expected to help resolve the long-standing open questions of minimal KS proofs 
%(with the smallest number of atomic propositions) 
in a given Hilbert space dimension.

\section{Results}
Our results presented in this work are in two parts. On the one hand, as the fundamental results, we introduced the order $(m,k)$ gadgets which play a crucial role in contextuality proofs. On the other hand, from the practical point of view, we proved that gadgets serve as the optimal measurement structures for several contextuality applications.

\subsection{Order $(m,k)$ gadgets and forbidden value assignments}
In this work, we first introduced a general class of state-dependent contextuality proofs termed order $(m,k)$ gadgets, which go beyond the known $01$-gadgets \cite{Clifton93,Stairs,hardy1993nonlocality,Belinfante73, RRHP+20,CA96,APSS18}. An order $(m,k)$ gadget contains $m$ mutually non-orthogonal vectors with the property that at most $k$ vectors among them can be assigned value $1$ in any valid $\{0,1\}$-assignment. In Hilbert space of dimension $d$, under some special constructions, we showed that when $m=d,k=d-1$ an order $(m,k)$ gadget can be constructed by any set of arbitrary non-orthogonal vectors $\{|v_1 \rangle, \dots, |v_m \rangle \}$. With this statement and construction, the novel Kochen-Specker (KS) sets as well as the general state-independent contextuality (SI-C) sets in dimension $d$ can be constructed with the order $(k,k-1)$ gadgets as building blocks (for a fixed value $k$ with $2\leq k\leq d$). Apart from this, we also showed that order $(m,k)$ gadgets form the building blocks of every KS proof by identifying them as induced subgraphs within any arbitrary KS proof.

We discussed about a class of even more general measurement structure in which some specific $\{0,1\}$-assignments are forbidden, a result that may be of independent interest and application.
From this point of view, a $01$-gadget is a set of vectors where the assignment $\{(1,1)\}$ is forbidden on two given non-orthogonal vectors, an order $(m,k)$ gadget is a set of vectors and for a given independent set $I = \{v_1, \dots, v_m\}$ of it,  the assignments of the form $\{(\underbrace{1,\dots,1}_{\tiny{k+1}}, f(v_{k+2}),\dots, f(v_m)) \bigwedge \left( \texttt{permutations}\right)\}$ are forbidden ($f$ is any $\{0,1\}$-assignment function).
 And the KS sets are the ones that demonstrate the full forbidden value assignments $\{0,1\}^{|I|}$ on any interdependent set $I$. A natural question then arises - can a gadget be constructed for every forbidden value assignment set? We answered this question in the affirmative and demonstrated the construction process in a concrete step-by-step manner.

\subsection{Gadgets as optimal measurement structures to contextuality applications.}

\textbf{Entanglement-assisted advantage in Zero-error communication.}
Zero-error information theory is one of the most important applications of contextuality. 
Given a single use of a discrete, memoryless channel $\mathcal{N}$, the maximum number of (classical) messages that (a sender) Alice can send to (a receiver) Bob without causing any error is known as the \textit{one-shot zero-error capacity of} $\mathcal{N}$. 
In groundbreaking work, Cubitt, Leung, Matthews and Winter \cite{CLMW10} showed how to use KS proofs (specifically the KS graphs) to construct channels (confusability graphs) for which shared entanglement between Alice and Bob can increase the one-shot zero-error capacity.  In analogy with them, we took the confusability graphs to be the orthogonality graphs of a certain class of gadgets and we 
considered a weighted version of the problem in which we assign weights $w_i$ to the input symbols denoting the desirability of their transmission.
 It is in such a weighted version of the zero-error communication problem, we obtained an enhancement of the one-shot zero-error capacity via shared entanglement for channels corresponding to specific types of gadgets, which is a much wider class of graphs than was previously known. 

\textbf{Large Violations of Binary Consistent Correlations in Quantum Theory.}
While Quantum Theory is the most successful theory ever devised, there is still a huge research effort devoted to understanding physical and information-theoretic principles that force its formalism, one class of them is the Fundamentally Binary theories, these are no-signaling theories in which measurements yield many outcomes are constructed by selecting from binary measurements.  Previously, the authors in \cite{KC16, KVC17} showed a Bell-type inequality to exclude the set of fundamentally binary non-signalling correlations as an underlying mechanism generating the set of quantum correlations, however, the proof is experimentally demanding since only small violations of the derived inequalities are possible, the violation of the derived inequality requires visibilities of $\approx 91.7 \%$ of a suitably prepared two-qutrit state. 
We use the $01$-gadget structures to derive the inequalities which are actually the maximum fractional assignments sum of the distinguished vectors, and from which the genuinely ternary character of quantum measurements can be certified in the much simpler single-system contextuality scenario with arbitrarily large separations between the set of quantum contextual correlations and the binary consistent correlations.

\textbf{Optimal semi-device-independent randomness generation.}
Contextuality can serve as the basis for randomness (or key) generation, and importantly one may utilize gadget-based contextuality tests to certify the optimal amount of randomness $\log_2 d$ per run which is the maximum randomness that can be extracted from a system of dimension $d$. To do that, in general one needs to derive a rigid contextuality test in dimension $d$ and identify a suitable measurement $x^*$ with fully random outcomes $P_{A|X}(a|x^*) = 1/d \quad\forall a \in [d] $ when the maximum quantum value of the contextuality test is observed.
We showed that the general constructions in \cite{RRHP+20} and the tunability of the overlap between the distinguished vectors of the gadgets make them ideal candidates for protocols allowing to certify the maximum amount of $\log_2 d $ bits of randomness. Specifically, we demonstrated the $01$-gadgets in dimension $d=4$ and $5$ with the maximum overlaps between the distinguished vertices being $\frac{1}{\sqrt{d}}$ (for any orthogonal representation in $\mathbb{R}^d$), which indicate that one can readily derive contextuality tests allowing $\log_2 d$ bits optimal randomness certification from these constructions.

\section{Discussion}
 In this paper, we have introduced a generalisation of gadget structures to definite prediction sets for arbitrary
logical propositions and shown how gadgets are optimal measurement structures in many applications of contextuality. 
 A number of interesting open questions remain. A fundamental question is to leverage the constructions of gadgets and their utility in building KS proofs to identify minimal KS proofs in a given dimension. It is still an open question to identify the minimal KS proof in dimension $3$ while the $18$-vector set introduced in \cite{CEG96, Cab08} is conjectured to be minimal in dimension $4$. With regard to applications, it is of interest to construct minimal gadget sets of measurements giving a contextuality test to certify the optimal amount of $\log d$ bits from a system of dimension $d$ (constructions were shown for small dimensions here), and to use them in experimentally feasible contextuality-based randomness generation protocols. It is also of interest to experimentally test the separation between quantum mechanics and general binary consistent theories. In future, it would be interesting to see if gadget generalisations can be used to show separations between quantum correlations and the set of $n$-ary consistent correlations that are defined analogously to the binary theories as composed from measurements yielding at most $n$ outcomes.

\revise{\section{Methods}}

\subsection{Preliminaries}

Much of the reasoning involving outcome contextuality has traditionally been carried out using graph-theoretic representations of KS sets, we therefore begin by establishing some graph-theoretic notation.  

\textbf{Graphs}. In this paper, we deal with simple, undirected, finite graphs $G = (V_G, E_G)$ where $V_G$ and $E_G$ denote the vertex and edge set of the graph respectively. If two vertices $v_i, v_j$ are connected by an edge, we say that they are adjacent and denote it by $v_i \sim v_j$. A clique $C$ in the graph $G$ is a subset of vertices $C \subset V_G$ such that every pair of vertices in $C$ is connected by an edge. A maximal clique is a clique that is not a subset of a larger clique, while a maximum clique in $G$ is a clique of maximum size in $G$, we denote the size of the maximum clique by $\omega(G)$. 
%In general, there can be a number of cliques of maximum size in $G$. A simple albeit inefficient method to identify them is to consider all the $n \choose \omega(G)$ subsets of $V_G$ and check which subsets are cliques, we will use this simple fact later. 
An independent set in a graph $G$ is a subset $I\subseteq V_G$ such that every pair of vertices in $I$ is non-adjacent in $G$, the maximum size of an independent set is denoted $\alpha(G)$. A set of vertices $D\subset V_G$ is said to dominate a clique $C$ if for every $v \in C$ there exists a $w \in D$ such that $\{v,w\}\in E$, that is every vertex of clique $C$ has a neighbour in $D$. The set
$D$ is a minimal dominating set of $C$ if no proper subset of $D$ dominates $C$.

\textbf{Orthogonality graphs and orthogonal representations}. For any set of vectors $\mathcal{V}$, one can define an orthogonality graph $G_{\mathcal{V}}$ as the graph in which vector $|v \rangle \in \mathcal{V}$ is represented by a vertex $v$ in $G_{\mathcal{V}}$ and two vertices $v_1, v_2$ are connected by an edge in $G_{\mathcal{V}}$ if and only if $\langle v_1 | v_2 \rangle = 0$ \cite{Lovasz87}. Checking the orthogonality relations among the vectors from a given set ${\mathcal{V}}$ allows to efficiently establish its orthogonality graph $G_{\mathcal{V}}$. The problem of $\{0,1\}$-coloring of a given set of vectors is then equivalently formulated as the problem of $\{0,1\}$-coloring of its orthogonality graph defined in an analogous way as: 
%It can be also generalized
%to the problem of \zo-coloring of an arbitrary graph:
\begin{definition}
A $\{0,1\}$-coloring of a graph $G$ is a map 
$f:V_G\rightarrow \{0,1\}$ such that (i) for every clique $C$ in $G$, it holds that $\sum_{v \in C} f(v) \leq 1$, and (ii) for every clique $C$ in $G$ of size $\omega(G)$, there exists exactly one vertex $v \in C$ satisfying $f(v)=1$.
\end{definition}
The converse problem of identifying sets of vectors satisfying the orthogonality constraints dictated by the edges of a given graph is the question of finding an orthogonal representation of a graph. 
\begin{definition}
An orthogonal representation of a graph $G$ in dimension $d$ is a set of (unit) vectors $\mathcal{S}$ from $\mathbb{C}^d$ such that there exists a map $f:V_G \mapsto {\mathcal S}$ satisfying the condition that $f(v_1)$ and $f(v_2)$ are orthogonal vectors if $\{v_1, v_2\} \in E_G$. The minimal dimension of an orthogonal representation of $G$ is denoted $d(G)$. A faithful orthogonal representation of a graph $G$ in dimension $d$ is a set of (unit) vectors $\mathcal{S}$ from $\mathbb{C}^d$ such that there exists a map $f:V_G \mapsto \mathcal{S}$ satisfying the condition that $f(v_1)$ and $f(v_2)$ are orthogonal vectors if and only if $\{v_1, v_2\} \in E_G$, and furthermore $f(u)$ and $f(v)$ are non-parallel vectors if $u \neq v$. The minimal dimension of a faithful orthogonal representation if $G$ is denoted $d^*(G)$.

%if $uv\in E(G)$, then the vectors $f(u), f(v)$ are orthogonal. Minimal dimension $d$ with this property is denoted as $d(G)$ and we say that $G$ has dimension $d(G)$.

%A faithful orthogonal representation is any set of vectors from ${\mathcal C}^d$, for which there exists a map $f:V\mapsto {\mathcal S}$ such that $f(u), f(v)$ are orthogonal if and only if $uv\in E(G)$, and additionally 
%$u\neq v \iff f(u) \neq f(v)$. Minimal dimension $d$ of a faithful representation is denoted as $d^*(G)$ and we say that $G$ has faithful dimension $d^*(G)$.
\end{definition}

We recall here the notion of $01$-gadgets formalised in  \cite{RRHP+20}.
%the notion of 
%$01$-gadgets has been introduced. We recall
%here their definitions in terms of the set of vectors, and a one, equivalent in terms of a graph.
\begin{definition} \cite{RRHP+20}
A $01$-gadget in dimension $d$ is a \zo-colorable set ${\mathcal S}_{gad}\subset \mathbb{C}^{d}$ of vectors containing two
distinguished non-orthogonal vectors $|u\>$ and $|v\>$ that nevertheless satisfy $f(u)+f(v)\leq 1$ in every \zo-coloring $f$ of ${\mathcal S}_{gad}$.
\end{definition}
Equivalently, the $01$-gadgets may be defined in graph-theoretic terms as: 
%the alternative, equivalent definition reads:
\begin{definition} \cite{RRHP+20}
A $01$-gadget in dimension $d$ is a \zo-colorable   graph $G_{gad}$ with faithful
dimension $d^*(G_{gad})=\omega(G_{gad})=d$ and 
with two distinguished non-adjacent vertices $u$ and $v$ such that $f(u)+f(v)\leq 1$ in every \zo-coloring f of $G_{gad}$.
\end{definition}
In other words, $01$-gadgets are particular definite-prediction sets with a logical implication of the form $P \rightarrow \overline{Q}$, i.e., in any logical assignment of the set of atomic propositions, when one of the two distinguished propositions is assigned the value True the other is necessarily assigned value False, even though the distinguished atomic propositions are not represented by orthogonal vectors and are therefore not inherently exclusive to each other. In \cite{RRHP+20}, it was shown that $01$-gadgets identify the essential contradiction captured by the Kochen-Specker theorem, in that every KS graph contains a $01$-gadget and from every $01$-gadget one can construct a proof of the Kochen-Specker theorem (see also \cite{AM78, AM80}). 
 Note that by the famous Erd\H{o}s-Stone theorem \cite{Erdos-Stone} of extremal graph theory, graphs of sufficiently high density necessarily contain subgraphs isomorphic to $01$-gadgets, specifically the maximum number of edges in a graph (with faithful dimension $d$) with $n$ vertices not containing a subgraph isomorphic to a $01$-gadget (of dimension $d$) is $\left[\frac{d-2}{d-1} + o(1) \right] \binom{n}{2}$. This can be seen by observing that $01$-gadgets in dimension $d$ have chromatic number $\chi(G) = d$, where the chromatic number of a graph denotes the minimum number of colors needed to color the vertices such that adjacent vertices are assigned distinct colors.

\textbf{Kochen-Specker sets, $01$-gadgets and Satisfiability}. From the preceding discussion, we recognize that the orthogonality graphs of Kochen-Specker vector sets do not admit a $\{0,1\}$-coloring. The $\{0,1\}$-colorability of a graph can also be formulated as a SAT instance and solved using a solver such as MiniSAT (http://minisat.se). To do this, one introduces a variable for each vertex in the graph. For each edge in the graph, a clause is added stating that the two incident vertices cannot both have value $1$ (True). For each maximum clique in the graph, a clause is added stating that not all vertices in the clique have value $0$ (False). The Boolean formulas for KS graphs is then seen to be unsatisfiable. Specifically, for the dimension $d=3$ setting, one can formulate the $\{0,1\}$-colorability of KS graphs as a 1-in-3 SAT instance. To do this, we complete the bases in the KS set by adding appropriate (unique) vectors, such that each edge in the graph belongs to a triangle. The Boolean formula in conjunctive normal form then has exactly three literals per clause, i.e., the formula is of the form $\bigwedge_{(i_1,i_2,i_3) \in \text{Cliques}}(v_{i_1} \vee v_{i_2} \vee v_{i_3})$ and the $\{0,1\}$-colorability is equivalent to the 1-in-3 SAT question of determining whether there exists a truth assignment to the variables so that each clause has \textit{exactly} one true literal. Furthermore, one can also obtain a similar unsatisfiable formula for $01$-gadgets with an added clause stating that the two distinguished non-adjacent vertices both have value $1$. Thus, from the point of view of satisfiability, $01$-gadgets provide a similar (and in many cases, smaller) unsatisfiable instance. From the point of view of contextuality, $01$-gadgets provide a state-dependent version of Kochen-Specker contextuality.  
   
\subsection{Order $(m,k)$ gadgets}
Let us now consider generalisations of gadget measurement structures that go beyond the basic `true-implies-false' and `true-implies-true' logical implications. Our first generalisation is to gadgets of order $(m,k)$ with $k \leq m$. Essentially, these are prediction sets corresponding to the proposition $\left[\left(\bigwedge_{i=1}^k P_i \rightarrow \bigwedge_{j=k+1}^{m} \overline{P}_j \right) \bigwedge \left( \texttt{permutations} \right)\right]$ for $m$ mutually non-exclusive atomic propositions $P_1, \dots, P_m$. In other words, the gadgets of order $(m,k)$ contain $m$ mutually non-orthogonal vectors such that at most $k$ vectors can be assigned value $1$ in any $\{0,1\}$-coloring. The $01$-gadgets \cite{RRHP+20, Clifton93, KS90, Greechie-1, Pitowsky-1, Pitowsky-2}  then correspond to the special case of gadgets of order $(2,1)$. 
\begin{definition}
A gadget of order $(m,k)$ in dimension $d$ is a \revise{\zo-assignable} set of vectors ${\mathcal S}_{m,k} \subset \mathbb{C}^{d}$ containing $m$ distinguished mutually non-orthogonal vectors \revise{ $\widetilde{{\mathcal S}}_{m,k} = \{|v_1\>,...,|v_m\>\}\subsetneq {\mathcal S}_{m,k} $}, such that 
%for any valid \zo-coloring at most $k$ of them can be attributed value 1 simultaneously and any $k$ of them can be attributed value 1 simultaneously by a valid \zo-coloring.
\begin{itemize} 
    \item for every subset \revise{$\mathcal{R} \subset \widetilde{{\mathcal S}}_{m,k}$} of size \revise{smaller than or equal to $k$}, there exists a \zo-coloring which attributes $1$ to all vectors in $\mathcal{R}$, and
    \item for any subset \revise{$\mathcal{R} \subset \widetilde{\mathcal S}_{m,k}$} of size greater than $k$,  
    no $\{0,1\}$-coloring exists that attributes $1$ to all vectors in $\mathcal{R}$.
\end{itemize}
\end{definition}

One can also give an equivalent definition of the order $(m,k)$ gadget in graph-theoretic terms
%It will be also useful to define the higher order gadget
%in graph-theoretical terms.
\begin{definition}
A gadget of order $(m,k)$ in dimension $d$ is a $\{0,1\}$-colorable graph $G$ with faithful
dimension $d^*(G_{gad})=\omega(G_{gad})=d$ and with a distinguished independent set $I$ of cardinality $|I| = m$ such that
\begin{itemize}
    \item for every subset $I' \subset I$ of cardinality \revise{$|I'| \leq k$}, there exists a $\{0,1\}$-coloring of $G$ in which all $v \in I'$ are assigned value $1$, and
    \item no $\{0,1\}$-coloring of $G$ exists that assigns value $1$ to more than $k$ vertices from $I$. 
\end{itemize}
%$m$ distinguished mutually non-adjacent  vertices $V'=\{v_1,...,v_d\}\subsetneq V$, such that
%(i) for every subset $V''\subseteq V'$ of cardinality $|V''|\leq k$ there exists $01$-coloring of $G$ such that all $v_i\in V''$ are mapped to $1$ in this coloring. 
%(ii) there does not exist $01$-coloring of $G$
%such that more than $k$ vertices from $V'$ are mapped to $1$.
\label{def:hog}
\end{definition}

We first study the question whether a higher order $(m,k)$ gadget can be constructed with any set of arbitrary vectors $\{|v_1 \rangle, \dots, |v_m \rangle \}$ as the distinguished vectors. While it is possible to consider every value of $k \in [m-1]$, here we focus on the construction for the special case $m=d, k = d-1$ .
%and show the application of such a construction in the following section. 
As in the construction of KS sets, the construction of such general gadgets is complicated by the fact that even deciding the $\{0,1\}$-colorability of a general graph is an NP-complete problem \cite{Arends09}. It is also hard in general to derive the faithful orthogonal representation of a graph in a given dimension. As such, there isn't a systematic method to derive minimal gadget structures. Nevertheless, we propose specific graphs $G$ with candidate vertices to play the role of the distinguished vertices of the gadget. We then construct a symmetric matrix $\texttt{Gram}$ with entries $\texttt{Gram}_{i,j} = \texttt{Gram}_{j,i} = 0$ corresponding to edges $(i,j)$ in $G$. The matrix $\texttt{Gram}$ is meant to represent the Gram matrix of a set of vectors realising the graph $G$ so that $\texttt{Gram}_{i,j} = \langle v_i | v_j \rangle$. We study the question of finding a positive-semi-definite matrix completion $\texttt{Gram} \succeq 0$ with a rank-$d$ constraint.  
%In the Supplemental Material,
%Appendix \ref{speical_gagdet}, 
We thus exhibit a graph that serves as an order $(d,d-1)$ gadget for arbitrary $d$, with the $d$ distinguished vectors being $|m_1 \rangle, |m_2 \rangle,\cdots,|m_{d} \rangle$ (details are in the \revise{Supplementary Information Note 1}).
The feature of this construction is that the distinguished vectors can be chosen to be arbitrary close to each other, i.e., $\langle m_i | m_j \rangle \rightarrow 1$ as the number of repeating units increases.

We now show an application of the higher-order gadgets in constructing novel KS proofs as well as general state-independent contextual (\revise{SI-C}) proofs and also defer the details of these contractions to \revise{Supplementary Information Note 2.}
%that the construction in dimension $d$ of order $(k,k-1)$ gadgets (for any fixed value $k$ with $2 \leq k \leq d$) can then be used to efficiently build novel KS proofs as well as general state-independent contextual (SIC) proofs of the Yu-Oh type. 
%Fix a value of $k$ in the range $\{2,\dots, d\}$. 
\begin{construction}\label{cos_KS}
Order $(k,k-1)$ gadgets can be used as building blocks to construct KS proofs in dimension $d$.
\end{construction}
In the construction, we start with $k$ bases $B_1, B_2, \dots, B_k$ in dimension $d$, then randomly pick one vector in each basis to form a set $S_i = \big\{ |v_{B_p}^q \rangle \big\}$ with $p \in [k] := \{1,\dots,k\}$ and $q \in [d]$. In total, we have $d^k$ such sets $S_i$. Then for each $i \in [d^k]$, we construct an order $(k,k-1)$ gadget in dimension $d$ with the vectors in $S_i$ being the distinguished vectors. Thus, assigning a single value $1$ to each of the bases $B_1, \dots, B_{k-1}$ forces all the vectors in the basis $B_k$ to be assigned value $0$ giving a contradiction, so that the union of all vectors is a KS proof. 
% We defer the details of this contraction of Kochen Specker proofs to Supplemental Material.
%ppendix \ref{App:SIC-const}.
\begin{construction}
Order $(k,k-1)$ gadgets can be used as building blocks to construct general \revise{SI-C} sets in dimension $d$. 
\end{construction}
To realize the general \revise{SI-C} set, we first construct a set of $r \cdot 2^n$ distinct unit vectors $|u_i \rangle$ in dimension $d$ satisfying
$\sum_{i=1}^{r\cdot 2^n}|u_{i}\rangle\langle u_{i}|=\frac{r\cdot 2^n}{d} \mathbb{1}_d$, where $r> \max \left\{\frac{d(k-1)}{2^n}, 4 \right\}$ is an even integer and $n = \begin{cases} 
      \lceil\log_2 \frac{d-1}{2}\rceil, & $d$ \; \text{is odd} \\
      \lceil\log_2 \frac{d-2}{2}\rceil, & $d$ \; \text{is even}
   \end{cases}$. Then any $k$ of these vectors form a set $S_i$, we first delete all the mutually orthogonal vectors in the set $S_i$ and construct an order $(|S_i|, |S_i|-1)$ gadget in dimension $d$ with the vectors in $S_i$ being the distinguished vectors. As a result, in any $\{0,1\}$-assignment $f$, the sum of assignments of these $r \cdot 2^n$ vectors is smaller than $k$. On the other hand, in quantum theory we obtain the value $\frac{r\cdot 2^n}{d} > k$ for every state in dimension $d$, so that the union of all the vectors gives a proof of state-independent contextuality.
% \begin{theorem}\label{state_ind}
% Set 
% %\begin{equation}
% $n = \begin{cases} 
%       \lceil\log_2 \frac{d-1}{2}\rceil, & $d$ \; \text{is odd} \\
%       \lceil\log_2 \frac{d-2}{2}\rceil, & $d$ \; \text{is even}
%   \end{cases}.$
% %\end{equation}
% Let $r \geq 4$ be an even integer. There exist $r \cdot 2^n$ distinct unit vectors $|u_i \rangle$ in dimension $d$ satisfying
% \begin{equation}
% \sum_{i=1}^{r\cdot 2^n}|u_{i}\rangle\langle u_{i}|=\frac{r\cdot 2^n}{d} \mathbb{1}_d
% \end{equation}
% \end{theorem}
% We defer the construction of such a set of vectors and the proof of this theorem, and details of the SIC sets construction to Supplemental Material.
% %Appendix \ref{App:SIC-const}.
%In \cite{RRHP+20}, 
%we identified the importance of the order $(2,1)$ gadgets (termed $01$-gadgets or `true-implies-false' sets or bugs in the literature) in Kochen-Specker proofs. Specifically, 
%we showed how  the order $(2,1)$ gadgets capture the essential contradiction necessary to prove the KS theorem. 
%i.e., every Kochen-Specker graph contains such an order $(2,1)$ gadget and using every $(2,1)$ gadget one can build a proof of the Kochen-Specker theorem. 
Finally, not only can the higher-order gadgets be used as building blocks to construct KS proofs, we also show that specific such gadgets may be found as necessary substructures (induced subgraphs) in any proof of the KS theorem.
%A natural question arises as to the role of the higher-order gadgets introduced in this paper in proofs of the Kochen-Specker theorem. 
%In the following, we show how specific higher-order gadgets may also be found as substructures (induced subgraphs) in any proof of the Kochen-Specker theorem. 
\begin{theorem}\label{thm:kk-1gadget_in_KS}
Every KS set in dimension $d$ contains a gadget of order $(k,k-1)$ for some $k$ satisfying $2\leq k \leq d$.
\end{theorem}
The intuition behind the proof is that if no $\{0,1\}$-coloring exists for a graph $G$, a brute-force greedy algorithm that attempts to assign $0$s and $1$s to its vertices must stop at some point in its execution, before each maximum clique has a single $1$-valued vertex. Therefore, there must exist some clique $C$ in $G$ such that each vertex in $C$ is adjacent to some $1$-valued vertex at this point in the execution. Call such a minimal set of adjacent vertices to a maximum clique as $D$, then the induced subgraph formed by $D \cup C$ constitutes a gadget. Furthermore, such a gadget must be of order at least $(k,k-1)$. 
%By Theorem \ref{thm:kk-1gadget_in_KS}, and the fact that $\omega(G) = d$ for a KS set in dimension $d$, it follows that every KS set in dimension $d$ contains a gadget of order $(k,k-1)$ for some $k \in [2, d]$. 

Order $(m,k)$ gadgets are thus a natural generalisation of 'True-implies-False' sets, where we consider an independent set $I = \{v_1, \dots, v_m\}$ of vertices (mutually non-orthogonal vectors) with forbidden value assignments of the form $\{(\underbrace{1,\dots,1}_{\tiny{k+1}}, f(v_{k+2}),\dots, f(v_m)\}$ and permutations thereof, for any $\{0,1\}$-coloring $f$. One may consider yet more general structures in which we specify a general set of forbidden value assignments $\mathcal{H} \subset \{0,1\}^m$, we elaborate on this in \revise{Supplementary Information Note 3.}

\subsection{Entanglement-assisted advantage in Zero-error communication in channels constructed from gadgets}
One of the most important and tantalizing applications of contextuality is in the field of zero-error information theory. In classical zero-error coding, we consider a discrete, memoryless channel $\mathcal{N}$ connecting a sender Alice and receiver Bob. Given a single use of such a channel, the maximum number of classical messages that Alice can send to Bob under the constraint that there be no error is known as the \textit{one-shot zero-error capacity} of $\mathcal{N}$. In groundbreaking work, Cubitt et al. \cite{CLMW10} showed how to use KS proofs to construct channels for which shared entanglement between Alice and Bob can increase the one-shot zero-error capacity. Since $01$-gadgets are substructures of KS proofs, it is an interesting question to investigate whether these smaller (and experimentally more feasible) measurement structures already exhibit the phenomenon of entanglement-assisted advantage in zero-error capacity. 

The classical channel $\mathcal{N}$ has  finite inputs $X$ and outputs $Y$ and its behavior is characterized by the probability distribution $\mathcal{N}_{Y|x}(y|x)$ of outputs given inputs. Two inputs $x$ are confusable if the corresponding distributions on outputs overlap. The confusability graph $G(\mathcal{N})$ is constructed with vertex set being the set of input symbols and two vertices connected by an edge if the corresponding input symbols are confusable. A zero-error code is then a set of non-confusable inputs and the one-shot zero-error capacity of the channel is  the maximum size of such a set. When Alice and Bob only share correlations which can be obtained using shared randomness, this number can be readily seen to be the independence number of the graph $G(\mathcal{N})$, i.e., $c_{SR}(\mathcal{N}) = \alpha\left(G(\mathcal{N})\right)$ where $c_{SR}(\mathcal{N})$ denotes the zero-error capacity when using correlations obtained using shared randomness as a resource. On the other hand, Cubitt et al. showed examples of channels for which sharing entanglement can improve the zero-error capacity of sending classical messages, i.e., such that $c_{SE}(\mathcal{N}) > c_{SR}(\mathcal{N})$ where $c_{SE}(\mathcal{N})$ denotes the zero-error capacity when using shared entanglement as a resource. In particular, they showed that such channels arise naturally from proofs of the Kochen-Specker theorem, specifically one may take $G(\mathcal{N})$ to be the (non-$\{0,1\}$-colorable) orthogonality graph of some Kochen-Specker vector set.  

We show that one may also take $G(\mathcal{N})$ to be the orthogonality graph of a certain class of gadgets, by a suitable generalisation to a weighted version of the zero-error communication problem.
In the weighted generalisation, we assign weights $w = \{w_i\}_{i=1}^{|V|}$ to the input symbols (denoting the desirability of their transmission). 
%While the zero-error code still remains a set of non-confusable inputs, 
The one-shot zero-error capacity is then the maximum total weight of any set of non-confusable inputs,
%In this case, the one-shot zero-error capacity of the channel is
which corresponds to the weighted independence number of the confusability graph, i.e., $c_{SR}(\mathcal{N}, w) = \alpha\left(G(\mathcal{N}), w \right)$. 

Consider a gadget in which we complete each of the bases (by addition of suitable vectors satisfying the orthogonality relations) such that a clique cover of the graph is possible in which the vertices of the graph are partitioned into $q$ maximum cliques (of size $\omega(G) = d$) given as $C_m = \{v_{m,1}, \dots, v_{m,d}\}$ for $m=1,\dots, q$ (i.e., $V = \cup_{m=1}^q C_m$). 
%An example of such a completion for a gadget is shown in Fig. \ref{}. 
We remark that a similar completion is required for the graphs obtained from Kochen-Specker proofs in \cite{CLMW10}, and only such Kochen-Specker proofs (such as the Peres-Mermin proof \cite{Mermin} with $24$ vectors partitioned into six cliques in dimension $4$) display the enhancement proven there. 

We construct the channel $\mathcal{N}$ as having inputs in $[q] \times [d]$ with inputs $(m,i)$ and $(m',i')$ being confusable if and only if the corresponding vectors are orthogonal to each other, i.e., if and only if $\langle v_{m,i} | v_{m',i'} \rangle = 0$. $G(\mathcal{N})$ has an edge between such pairs of confusable inputs and is exactly the orthogonality graph corresponding to the (base-completed) gadget. By construction, the vertices of $G(\mathcal{N})$ can be partitioned into $q$ maximum cliques (of size $d$). We now consider the weighted version of the zero-error communication problem with $V_{\text{dist}}$ denoting the set of distinguished vertices in the gadget as 

\begin{equation} 
\label{eq:weight-zec}
w_i = \left\{ \begin{array}{ll} 
      w^* & i \in V_{\text{dist}} \\
      1 & i \in V \setminus V_{\text{dist}} \\
   \end{array}
\right. 
\end{equation}
for a parameter $w^*$. The one-shot zero-error capacity when only shared randomness is available is then readily calculated to be $c_{SR}\left(G(\mathcal{N})\right) = \max\left\{ \alpha\left(G(\mathcal{N})\right) - 1 + w^*, \alpha\left(G(\mathcal{N})\right) - 3 + 2w^* \right\}$. We choose $w^*>1$ such that $2w^*-3 < w^*-1$, i.e., $1< w^* < 2$ giving $c_{SR}\left(G(\mathcal{N})\right) = \alpha\left(G(\mathcal{N})\right) - 1 + w^* < q +w^* - 1$. 

On the other hand, suppose Alice and Bob share a maximally entangled state $|\psi_d \rangle = \frac{1}{\sqrt{d}} \sum_{i=1}^{d} | i, i \rangle$. Each message $m$ that Alice wishes to send corresponds to a maximum clique in the aforementioned clique partitioning of the graph $G(\mathcal{N})$. To send $m$, Alice measures in the bases given by the clique $C_m$ and obtains an outcome $k \in [d]$ with probability $1/d$. Her input to the channel is then $(m,k)$. The output of the channel at Bob's end is one of the maximum cliques containing the vertex $v_{m,k}$ (not necessarily belonging to the clique partitioning of the graph). Bob performs a projective measurement corresponding to his received maximum clique, and his outcome reveals Alice's input to the channel. The one-shot zero-error capacity when shared entanglement is used as a resource is then calculated to be $c_{SE}\left(G(\mathcal{N})\right) = \frac{1}{d} \left[ q d - |V_{\text{dist}}| + |V_{\text{dist}}| \cdot  w^* \right] = q + \frac{(w^*-1)|V_{\text{dist}}|}{d}$. We see that $c_{SE}\left(G(\mathcal{N})\right) > c_{SR}\left(G(\mathcal{N})\right)$ whenever $|V_{\text{dist}}| > d$, i.e., whenever we have a gadget-type graph with $|V_{\text{dist}}|$ distinguished vertices of which only one can be assigned value $1$ in any non-contextual $\{0,1\}$ value assignment. 

We note that such a gadget-type graph does not correspond to a Kochen-Specker proof since it is $\{0,1\}$-colorable. On the other hand, one can construct a state-independent non-contextuality inequality for the graph that is violated by all states in dimension $d$, namely
%\begin{equation}
    $\sum_{v_i \in V_{\text{dist}}} P(e_{v_i}) \leq 1$,
%\end{equation}
where $P(e_i)$ refers to the probability of the event $e_{v_i}$ corresponding to the distinguished vertex $v_i$. Such graphs may therefore be said to be of the type discovered by Yu and Oh in \cite{YO12}, namely they exhibit state-independent contextuality despite not corresponding to a Kochen-Specker proof. And as we have seen, we obtain an enhancement via entanglement of the one-shot zero-error capacity for all such graphs, a much wider (and easily constructable following the constructions in \cite{RRHP+20} and Construction 2 in this work) class of graphs than was previously known. 

\subsection{Large Violations of Binary Consistent Correlations in Quantum Theory}

In this section, we describe a novel application of the gadget constructions to the task of excluding a natural alternative to quantum theory, namely the so-called "Fundamentally Binary theories" \cite{KC16, KVC17}. While Quantum Theory is the most successful theory ever devised, there is still a huge research effort devoted to understanding physical and information-theoretic principles that force its formalism. Seemingly natural alternatives to the set of correlations allowed by Quantum Theory exist such as the so-called `Almost Quantum' correlation set \cite{Navasc15}. Another class of natural alternatives is given by the Fundamentally Binary theories, these are no-signalling theories in which measurements yielding many outcomes are constructed by selecting from binary measurements. In other words, these theories posit that on a fundamental level only measurements with two outcomes exist, and scenarios where a measurement has more than two outcomes are achieved by classical post-processing of one or more two-outcome measurements. Fundamentally binary correlations are characterised as the convex hull of all consistent correlations $\{ P(a|x)\}$ obeying the constraint that for all $x$, it holds that $P(a|x) = 0$ for all but two outcomes $a$.

In \cite{KC16, KVC17}, it was shown that two-party non-locality scenarios exist such that the corresponding class of fundamentally binary non-signalling correlations does not fully encompass the set of quantum correlation. In other words, it was shown that a Bell-type inequality can be constructed to exclude the set of fundamentally binary non-signalling correlations as an underlying mechanism generating the set of quantum correlations. The authors of \cite{KC16, KVC17} considered simplest non-trivial polytope of fundamentally binary non-signalling correlations involving two parties that perform two measurements with three outcomes each. They computed the facets of the polytope using Fourier-Motzkin elimination using the software $\texttt{porta}$ and calculated the corresponding quantum violations using the NPA semidefinite programming hierarchy \cite{navascues2008convergent}. While an important foundational result, the proof in \cite{KVC17} is experimentally demanding in that only small violations of the derived inequalities are possible (the quantum value being $I_a = 2(2/3)^{3/2} \approx 1.0887$ compared to the value in binary theories of $I_a = 1$), the violation of the derived inequality requires visibilities of $\approx 91.7 \%$ of a suitably prepared two-qutrit state. In this section, we show that the genuinely ternary character of quantum measurements can be certified in the much simpler single-system contextuality scenario with arbitrarily large separations between the set of quantum contextual correlations and the binary consistent correlations. The price to pay for such large violations is the assumption, common to all contextuality experiments, that the same projector is measured in different contexts. 

Consider an orthogonality graph $G = (V_G, E_G)$ with a set of maximum cliques (contexts) $\mathcal{C} = \{A_1, \dots, A_k\}$ where each clique $A_i$ is of size $\omega(G) = d$. A box $B = \{P(a|x)\}$ is a set of conditional probability distributions with input $x \in \{1, \dots, k\}$ and output $a \in \{1, \dots, d\}$. A box is said to be compatible with an orthogonality graph $G$ if it is a family of (normalized) probability distributions such that for each $c \in \{A_1, \dots, A_k\}$, there is a corresponding probability distribution in this family. 
\begin{definition}
For a given orthogonality graph $G = (V_G, E_G)$ with a set of contexts $\mathcal{C}_G = \{A_1, \dots, A_k\}$, a box $B = \{P(a|x\}$ is said to be a Consistent Box if for all pairs $c, c' \in \mathcal{C}_G$ and for sets of vertices (projectors) $S_{c,c'} = c \cap c' \neq \emptyset$, it holds that

\begin{equation}
 \forall s \in S_{c,c'} \; \; \; \;   P(a = s | x = c) = P(a = s | x = c'). 
\end{equation}
The set of all consistent boxes $B$ compatible with an orthogonality graph $G$ is denoted by $\mathtt{B}^{c}_G$. 
\end{definition}
Note that the set of non-signalling boxes is a special case of such consistent boxes. 

Fundamentally binary correlations are a sub-class of consistent correlations obtained as the convex hull of consistent boxes for which for each context $c$ in the graph $G$ (maximum clique of size $\omega(G) = d$) at most two projectors (vertices in the clique) are assigned non-zero values that sum to unity and the remaining projectors in the context are assigned value $0$, together with any box obtained by local classical postprocessing of such boxes. Note that in each extremal binary consistent box, the assignment of values to the projectors is done in a consistent manner, so that the value assigned to any projector is independent of the context in which it is measured. Formally we define binary consistent correlations as follows.
\begin{definition}
For a given orthogonality graph $G = (V_G, E_G)$ with a set of contexts $\mathcal{C}_G = \{A_1, \dots, A_k\}$, a binary consistent assignment is a function $f : V_G \rightarrow [0,1]$ such that $\forall c \in \mathcal{C}_G$, $exists \; v_1, v_2 \in c$ such that $f(v_1) + f(v_2) = 1$ and $f(v_i) = 0$ for all $v_i \in c \setminus \{v_1, v_2\}$. Define the set of boxes $\mathtt{B}^{\text{bin-cons}}_G$ as the convex hull of boxes obtained by binary consistent assignments, i.e., 
%\begin{widetext}
\begin{eqnarray}
\mathtt{B}^{\text{bin-cons}}_G := &&\text{conv}\bigg\{ \{P(a|x)\} \in \mathtt{B}^{c}_G \; | \; \forall c \in \mathcal{C}_G, \; \exists s_1, s_2 \in c \; \nonumber \\
&&\text{s.t.} \; P(a = s_1 | x = c) + P(a = s_2 | x = c) = 1 \bigg\}.
\end{eqnarray}
%\end{widetext}
The set of Fundamentally Binary boxes $\mathtt{B}^{\text{bin}}_G$ is defined as the set of boxes that can be obtained by local classical postprocessing from any $B \in \mathtt{B}^{\text{bin-cons}}_G$. 
\end{definition}

We now show that not only does the set of Fundamentally Binary boxes not encompass the set of quantum contextual correlations, but that in fact there exist separating inequalities for which large violations by quantum contextual correlations can be obtained. 
%To do so, we will use constructions of 'extended $01$-gadgets' that some of us introduced in \cite{}. 
%\begin{definition}
%An extended $01$-gadget in dimension $d$ is a $\{0,1\}$-colorable graph $G_{xgad} = (V_{xgad}, E_{xgad})$ with faithful dimension $d^*(G_{xgad}) = \omega(G_{xgad}) = d$ and with two distinguished non-adjacent vertices $v_1 \nsim v_2$ such that in any assignment $f : V_{xgad} \rightarrow [0,1]$, it holds that $f(v_1) + f(v_2) < 2$. 
%\end{definition}
%We make use of the extended $01$-gadgets to construct separating inequalities of the set of fundamentally binary correlations which admit large quantum violations.

\begin{theorem}
There exist inequalities bounding the set of fundamentally binary consistent correlations that admit close to algebraic violations in quantum theory.
\end{theorem}
\begin{proof}
The proof will make use of the idea of `extended $01$-gadgets' that we introduced in \cite{RRHP+20}. 
\begin{definition}
\label{def:ext-gad}
An extended $01$-gadget in dimension $d$ is a $\{0,1\}$-colorable graph $G_{xgad} = (V_{xgad}, E_{xgad})$ with faithful dimension $d^*(G_{xgad}) = \omega(G_{xgad}) = d$ and with two distinguished non-adjacent vertices $v_1 \nsim v_2$ such that in any assignment $f : V_{xgad} \rightarrow [0,1]$, it holds that $f(v_1) + f(v_2) < 2$. 
\end{definition}
In other words, an extended $01$-gadget is similar to a normal $01$-gadget except that the defining characteristic holds for arbitrary assignments in $[0,1]$ rather than only to $\{0,1\}$ assignments.

In \cite{RRHP+20}, we had proven the following statement that shows a construction of an extended $01$-gadget between any two non-orthogonal vectors in $\mathbb{C}^d$. 
\begin{lemma*}[Theorem $4$ in \cite{RRHP+20}]
Let $|v_1 \rangle$ and $|v_2 \rangle$ be any two distinct non-orthogonal vectors in $\mathbb{C}^d$ with $d \geq 3$. Then there exists an orthogonality graph $G_{xgad}$ that constitutes an extended $01$-gadget in dimension $d$ with the corresponding vertices $v_1$ and $v_2$ being the distinguished vertices.
\end{lemma*}

We now show that for any extended $01$-gadget, the sum of the binary consistent (probability) assignments to the two distinguished vertices in any box $B \in \mathtt{B}_G^{\text{bin-cons}}$ is at most $3/2$. To do so, we recall the notion of the Fractional Stable-Set Polytope $(FSTAB(G))$ of a graph $G = (V_G, E_g)$ which is defined as

\begin{equation}
    FSTAB(G) = \big\{ \vec{x} \in \mathbb{R}_+^{|V_G|} \; | \; x_{v} + x_{w} \leq 1 \; \; \forall (v,w) \in E_G \big\}.
\end{equation}
We recognise that the fractional stable-set polytope is defined by similar constraints to the set of binary boxes $\mathtt{B}^{\text{bin-cons}}_G$ except for the fact that the defining constraint $x_v + x_w \leq 1$ in $FSTAB(G)$ is replaced by the constraint that $\forall c \in \mathcal{C}_G$, $\exists v, w \in c$ such that $x_v + x_w = 1$ in $\mathtt{B}^{\text{bin-cons}}_G$. 
By introducing a slack variable $y_{v,w}$ for each edge constraint, we rewrite the fractional stable-set polytope as 

\begin{equation}
    FSTAB(G) = \big\{ (\vec{x}, \vec{y}) \in \mathbb{R}_+^{|V_G|} \times \mathbb{R}_+^{|E_G|} \; | \; x_v + x_w + y_{v,w} = 1 \; \; \forall (v,w) \in E_G \big\}.
\end{equation}
Here, every vertex of $G$ indexes an $x$ variable while every edge of $G$ indexes a slack $y$ variable, so that $x$ and $y$ can be termed vertex variables and edge variables respectively. A vertex $v$ is said to be $k$-valued in the fractional assignment $(\vec{x}, \vec{y})$ if the corresponding vertex variable takes value $k$, and similarly an edge $(v,w)$ is said to be $j$-valued in the assignment if the corresponding edge variable takes value $j$. The following theorem by Nemhauser and Trotter \cite{NT75}, following an earlier result by Balinski \cite{Bal65} provides a characterisation of the vertices of $FSTAB(G)$. 
\begin{theorem}[Balinski \cite{Bal65}, Nemhauser and Trotter \cite{NT75}]
Let $\vec{x} \in \mathbb{R}_+^{|V_G|}$ be a vertex of $FSTAB(G)$. Then for every vertex $v \in V_G$, it holds that $x_v \in \{0, \frac{1}{2}, 1\}$, i.e., that every vertex is $0$ or $1/2$ or $1$-valued in $\vec{x}$. 
\end{theorem}
Now, since the set of normalization conditions $NORM_G := \{\text{NORM}_G^{c}\}_c$ for the maximum cliques $c \in \mathcal{C}_G$ where

\begin{equation}
    \text{NORM}_G^{c} := \{ \exists v, w \in c \; \text{s.t.} \; f(v) + f(w) = 1 \}
\end{equation}
form supporting hyperplanes of $FSTAB(G)$, we see that the vertices of $FSTAB(G) \cap NORM_G$ inherit the characterisation derived in the above theorem, i.e., the corresponding edge (slack) variables $y$ take value $0$ for each edge in the graph. We thus obtain
\begin{corollary}
Let $\{P(a|x)\}$ be a vertex of $\mathtt{B}_G^{\text{bin-cons}}$. Then for every context $c \in \mathcal{C}_G$ and for every outcome $s \in [d]$, it holds that $P(a=s|x = c) \in \{0, \frac{1}{2}, 1\}$.
\end{corollary}
It is also worth remarking that classical processing does not change the above property so that it holds also for the extreme points of the polytope of Fundamentally Binary boxes $\mathtt{B}_G^{\text{bin}}$.
Applying the above corollary to any orthogonality graph $G_{xgad}$ that constitutes an extended $01$-gadget, we see that the sum of the binary consistent assignments to the two distinguished vertices is at most $3/2$.

This statement, in conjunction with the constructions of extended $01$-gadgets in the lemma for distinguished vectors $|v_1 \rangle, |v_2 \rangle$ satisfying $|\langle v_1 | v_2 \rangle| \rightarrow 1$ shows that the inequality $P(a = v_1 | x = c_{v_1}) + P(a = v_2 | x = c_{v_2}) \leq 3/2$ forms a supporting inequality for $\mathtt{B}_{G_{xgad}}^{\text{bin}}$, where $c_{v_1}$ and $c_{v_2}$ are two contexts containing the vertices $v_1$ and $v_2$ respectively. On the other hand, measurements of the contexts on the state $|v \rangle = \frac{1}{\sqrt{2(1+\cos{\theta})}}(|v_1 \rangle + |v_2 \rangle)$ with $\langle v_i | v_j \rangle = \cos{\theta}$ show that Quantum Theory achieves the value $\left(1 + \cos{\theta}\right) \rightarrow 2$ as $\theta \rightarrow 0$.

\end{proof}

It is worth remarking that separations between the sets of binary consistent correlations and quantum correlations are not achieved by considering the inequalities for the usual Kochen-Specker proofs since both the sets achieve the algebraic value for those inequalities. This shows the importance of the constructions of gadgets and extended gadgets in deriving such separating hyperplanes. Furthermore, it is also clear that higher-order extended gadgets can be constructed in analogy with the constructions of higher-order gadgets in the rest of this paper. In the future, it would be interesting to see if these constructions can be used to show large separations between the set of quantum contextual correlations and the set of $n$-ary consistent correlations that are defined analogously to the binary consistent correlations as composed from measurements yielding at most $n$ outcomes. 
\subsection{Optimal semi-device-independent randomness generation using gadgets}
Contextuality can serve as the basis for randomness (or key) generation, either via stand-alone protocols that test for the violation of a non-contextuality inequality \cite{U13} (where one assumes that the measurements conform to the specific orthogonality graph), or through the conversion of a single-party contextuality test into a two-party Bell inequality \cite{HHHH+10} or through the conversion of a non-contextuality inequality to a prepare-and-measure protocol \cite{gupta2022quantum}. A common step in all such protocols \cite{PironioNature, RBHH+15, CR12, KAF17, BRGH+16, GHH+14} is the identification of a suitable measurement in the (contextuality) test that yields the highest possible randomness or key generation rate. It is well-known that the maximum randomness (quantified by the min-entropy) per run that can be extracted from a test where the parties perform projective measurements on a system of dimension $d$ is $\log_2 d$. The importance and utility of gadgets for randomness certification has been commented on previously, we elaborate on this aspect and focus on their importance for optimal randomness certification in this section. 

The Kochen-Specker theorem shows that it is impossible to assign classical (deterministic) values to all quantum observables in a consistent manner, i.e., independent of the context in which the observables are measured. However, as pointed out in \cite{ACS15, ACS14}, the fact that not all quantum observables can be assigned definite values does not imply that no observable can be assigned a definite outcome. And in general, proofs of contextuality do not specify which observables are value-indefinite. Specifically, for a contextuality test with a set of observables $\{A_1, \dots, A_k\}$ we want to solve

\begin{equation}
  \begin{split}
    \text{max} \; &P_{guess}(A_i|E)\\
     s.t.\;~%&\Gamma \geq 0,\\
           &I(P_{A|X}) = I^*,\\
           &P_{A,E|X} \in \mathcal{Q},
  \end{split}
\end{equation}
where $I(P_{A|X})$ is a non-contextuality inequality evaluated on the observed conditional probability distributions $P_{A|X}$, $I^* \in (I_c, I_q]$ with classical and quantum values given by $I_c$ and $I_q$ respectively, and $\mathcal{Q}$ denotes the set of conditional distributions (boxes) achievable by performing measurements (compatible with the test structure on Alice's side) on quantum states shared between Alice and adversary Eve, $P_{guess}(A_i|E) = \sum_{e} P(e) P_e(a=e|i) $ is the guessing probability of Alice's outcome by an adversary $E$. By an optimal rigid contextuality test in dimension $d$ we mean one in which there exist a measurement bases $x^*$ such that $P_{A|X}(a|x^*) = 1/d$ for all outcomes $a \in [d]$ when the maximum value $I_q$ is observed. It is an open question to derive such a rigid class of contextuality tests for arbitrary dimension $d$ (see for example \cite{gupta2022quantum} where the guessing probability was calculated for the well known $5$-cycle non-contextuality inequality \cite{KCBS08}). 
%and thus far a semidefinite programming relaxation has been needed 
%as an example, 
%the well known $5$-cycle non-contextuality inequality \cite{KCBS08} was found to yield a randomness of  $\approx 0.77$ bits per measurement run in \cite{gupta2022quantum}. 

Constructions of gadgets provide a candidate solution to the problem. Specifically, the extended $01$-gadgets from Definition \ref{def:ext-gad} were used in \cite{RRHP+20} as building blocks to construct sets of vectors $\mathcal{S}'$ such that for any $[0,1]$-assignment $f: \mathcal{S}' \rightarrow [0,1]$ it holds that $f(|v_1 \rangle), f(|v_2 \rangle) \in \{0,1\}$ if and only if $f(|v_1 \rangle) = f(|v_2 \rangle) = 0$. A first interesting aspect of these gadgets for randomness certification is that they allow to localise the randomness guaranteed by the KS theorem (note that a similar theorem with a more complicated construction was explored in \cite{hrushovski2004generalizations}). In other words, the observation in the contextuality test of $P(|v_1 \rangle) = 1$ guarantees that $0 < P(|v_2 \rangle) < 1$ for any consistent box $P$ compatible with the measurement structure of the gadget. Secondly, if one has a rigid construction \cite{BRV+19} with overlap $|\langle v_1 | v_2 \rangle| = 1/\sqrt{d}$, one can readily derive a contextuality test allowing optimal randomness certification (for example with a non-contextuality inequality of the form $\beta P(|v_1 \rangle) + P(|v_2 \rangle) \leq 1$ with $\beta \gg 1$, for which the optimal quantum value is then $\beta + 1/d$). Here, by a rigid construction we mean one for which there exists a non-contextuality inequality whose maximum violation certifies a fully random outcome (with uniform probabilities $1/d$) for one of the measurement bases in the construction. One way to ensure this is if for the construction, the set of vectors realizing its orthogonality graph $G$ is unique in $\mathbb{C}^{\omega(G)}$ (up to unitaries). For the gadget-within-gadget construction in the proof of Theorem $4$ in \cite{RRHP+20}, it was shown that for the $k$-th iteration in the construction the maximum overlap of the distinguished vectors takes the form $\frac{k}{k+2}$, so that the  construction allows optimal randomness certification for $d = 4$ at $k=2$. As shown in Fig. 1, the maximum overlap between the distinguished vectors $|u_1^{(2)} \rangle, |u_{8}^{(2)} \rangle$ is $\frac{1}{2}$, and the orthogonal representation of this gadget is
$\langle v_1^{(1)}|=\langle v_4^{(2)}| = (1,0,0)$,
$\langle v_2^{(1)}| = (0,-\frac{\sqrt{3}}{2},\frac{1}{2})$,
$\langle v_3^{(1)}| = (0,\frac{\sqrt{3}}{2},\frac{1}{2})$,
$\langle v_4^{(1)}| = (-1,-\frac{\sqrt{2}}{2},-\frac{\sqrt{6}}{2})$,
$\langle v_5^{(1)}| = (-1,-\frac{\sqrt{2}}{2},\frac{\sqrt{6}}{2})$,
$\langle v_6^{(1)}| = (\frac{\sqrt{2}}{3},-\frac{1}{6},-\frac{\sqrt{3}}{6})$,
$\langle v_7^{(1)}| = (\frac{\sqrt{2}}{3},-\frac{1}{6},\frac{\sqrt{3}}{6})$,
$\langle v_8^{(1)}| = \langle v_5^{(2)}|=\frac{2\sqrt{2}}{3}(\frac{\sqrt{2}}{4},1,0)$,
$\langle v_1^{(2)}| = (-\frac{3}{2},-\frac{3\sqrt{2}}{4},\frac{3\sqrt{2}}{4})$,
$\langle v_2^{(2)}| = (0,1,1)$,
$\langle v_3^{(2)}| = (-\frac{3\sqrt{2}}{4},\frac{3}{8},-\frac{8}{9})$,
$\langle v_6^{(2)}| = (0,-1,1)$,
$\langle v_7^{(2)}| = (1,-\frac{\sqrt{2}}{4},-\frac{3\sqrt{2}}{4})$
$\langle v_8^{(2)}| = (\sqrt{2},1,1)$.

% \begin{widetext}

\begin{figure}
    \centering
    \includegraphics[width=10cm]{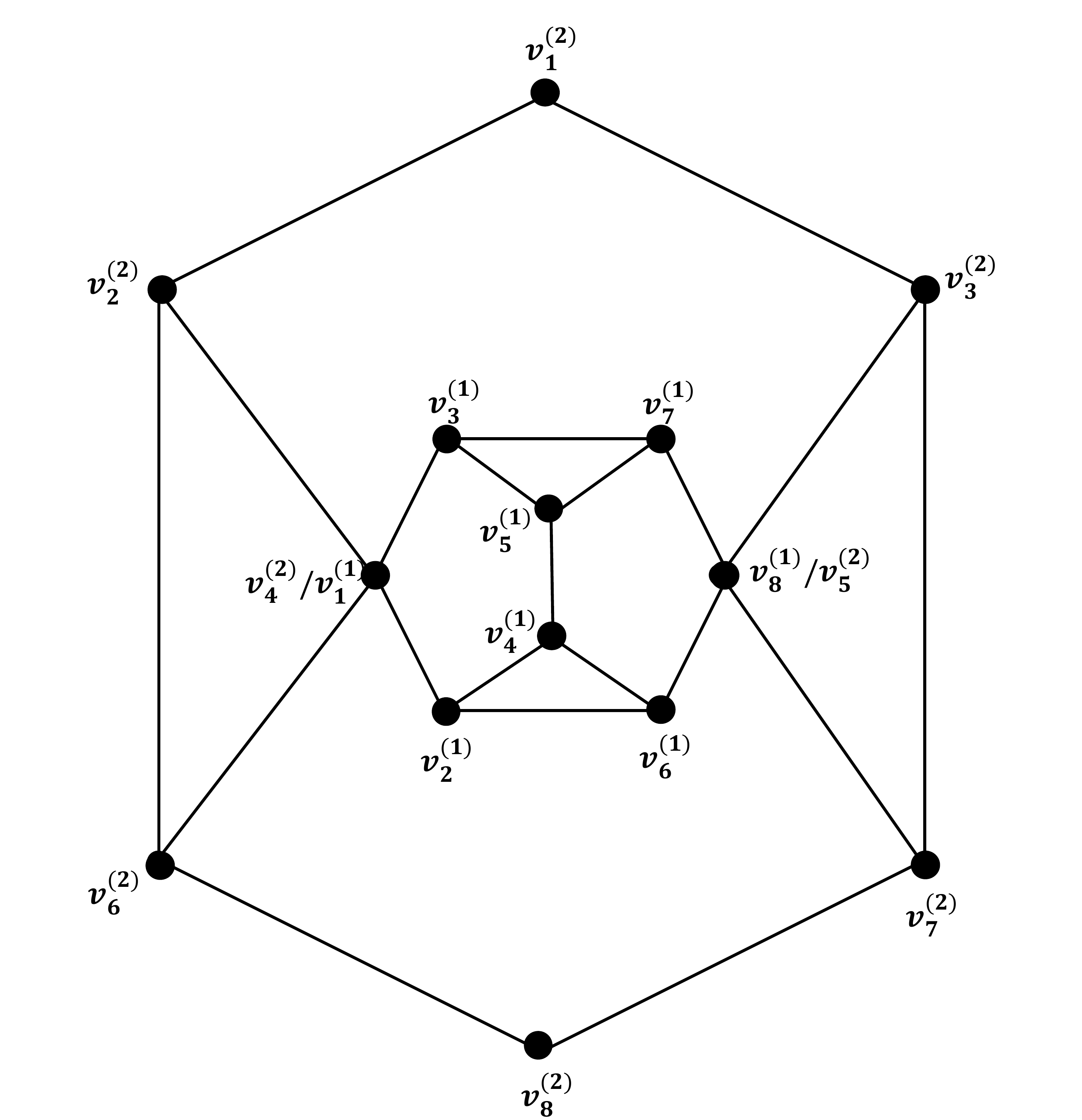}
    \caption{An example illustrating the construction of a gadget with prescribed maximum overlap between the distinguished vectors $|v_1^{(2)} \rangle, |v_{8}^{(2)} \rangle$ of $\frac{1}{2}$ for any real orthogonal representation. The example illustrates the utility of gadgets in deriving contextuality tests that allow for optimal randomness certification for $d=4$.}
    \label{ran4}
\end{figure}

\begin{figure}
    \centering
    \includegraphics[width=15cm]{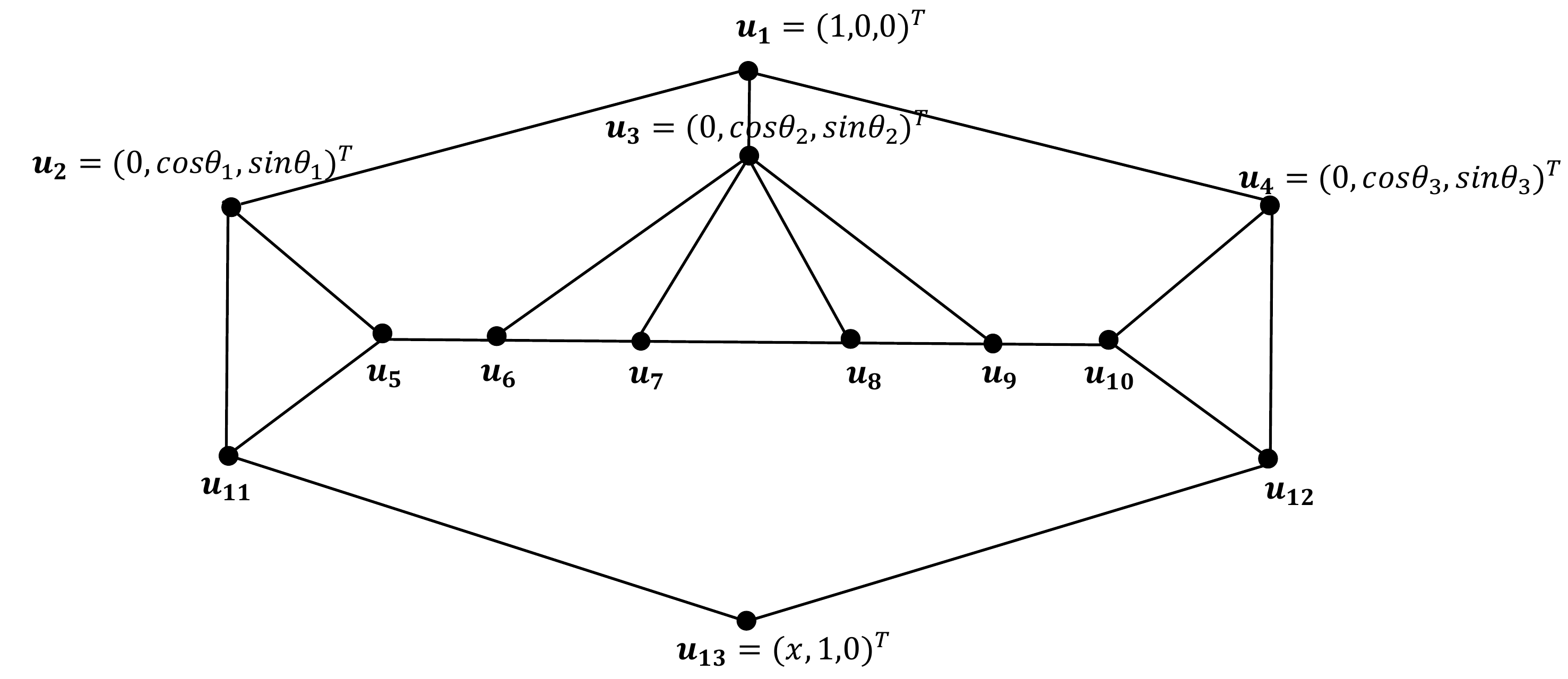}
    \caption{An example illustrating the construction of a gadget with prescribed maximum overlap between the distinguished vectors $|u_1 \rangle, |u_{13} \rangle$ of $\frac{1}{\sqrt{5}}$ for any real orthogonal representation. The example illustrates the utility of gadgets in deriving contextuality tests that allow for optimal randomness certification for $d=5$.}
    \label{ran}
\end{figure}
%\end{widetext}
We give a different construction that allows to certify $\log_2 d$ bits of randomness in dimension $d = 5$ here, and pursue the general question of rigid contextuality tests \cite{RLH12, BRV+19} certifying $\log_2 d$ bits for arbitrary $d$ (as well as their monogamy relations \cite{RSKK12, SR17} and the security proofs of the corresponding protocols) for future work.

Consider the orthogonality graph shown in Fig. 2. Without loss of generality, we consider $\langle u_1| = (1,0,0)$, $\langle u_{13} | = \frac{1}{\sqrt{1+x^2}}(x,1,0)$. Parametrizing $\langle u_2| = (0, \cos{\theta_1}, \sin{\theta_1})$, $\langle u_3| =  (0, \cos{\theta_2}, \sin{\theta_2})$ and $\langle u_4 | =  (0, \cos{\theta_4}, \sin{\theta_4})$ to ensure orthogonality with $\langle u_1|$, we deduce the following (unnormalized) vectors by taking appropriate cross products $\langle u_{11}| = (-\sin{\theta_1}, x \sin{\theta_1}, - x \cos{\theta_1})$, $\langle u_5| = (-x,-\sin^2{\theta_1},(1/2)\sin{2 \theta_1})$, $\langle u_6| = (-\cos{(\theta_1 - \theta_2)} \sin{\theta_1}, x \sin{\theta_2}, -x \cos{\theta_2})$, \\$\langle u_7| = (-x, -\cos{(\theta_1 - \theta_2)} \sin{\theta_1} \sin{\theta_2}, \cos{(\theta_1 - \theta_2)} \sin{\theta_1} \cos{\theta_2})$, $\langle u_{13} | = (-\sin{\theta_3}, x \sin{\theta_3}, - x \cos{\theta_3})$, $\langle u_{10} | = (-x, -\sin^2{\theta_3}, (1/2)\sin{2 \theta_3})$, $\langle u_9 | = (\cos{(\theta_2 - \theta_3)} \sin{\theta_3}, - x \sin{\theta_2}, x \cos{\theta_2})$ and $\langle u_8| = (x, \cos{(\theta_2 - \theta_3)} \sin{\theta_2} \sin{\theta_3}, -\cos{(\theta_2 - \theta_3)} \cos{\theta_2} \sin{\theta_3})$. We now ask what is the maximum value of the overlap $|\langle u_1| u_{13} \rangle| = \frac{x}{\sqrt{1+x^2}}$ under the constraint that \\$-\langle u_7|u_8 \rangle =  x^2 + \cos{(\theta_1 - \theta_2)} \cos{(\theta_3 - \theta_2)} \sin{\theta_1} \sin{\theta_3}$ equals $0$. Or equivalently we wish to minimize $\cos{(\theta_1 - \theta_2)} \cos{(\theta_2 - \theta_3)} \sin{\theta_1} \sin{\theta_3}$. Setting the partial derivatives of this expression with respect to $\theta_1, \theta_2, \theta_3$ to equal zero, and finding the maximum overlap over all the solutions gives that $\theta_1 = -\theta_3 = \pi/4$ and $\theta_2 = 0$ with the optimal overlap $|\langle u_1| u_{13} \rangle| = 1/\sqrt{5}$. The addition of two other vertices $u_{14}, u_{15}$ that are adjacent to all the vertices of the graph as well as to each other, gives the natural orthogonal representation in dimension $5$ with $\langle u_{14}| = (0,0,0,1,0)$ and $\langle u_{15} | = (0,0,0,0,1)$. We have thus constructed an extended $01$-gadget in dimension $5$ with the maximum overlap between the distinguished vertices being $\frac{1}{\sqrt{5}}$ (for any orthogonal representation in $\mathbb{R}^5$). While not a full self-testing statement, this indicates that the maximum violation of a non-contextuality inequality could allow to certify $\log_2 5$ bits for this construction.
As stated earlier, we leave for future work the derivation of rigid contextuality tests based on gadgets to certify $\log_2 d$ bits for arbitrary dimension $d$ and the security proofs of the corresponding (semi-device-independent) contextuality-based randomness generation.

\subsection{Classical Value of Non-Contextuality Inequalities versus the weighted independence number}
An interesting offshoot of the study of gadget structures is to point out a subtle modification in a famous result by Cabello, Severini and Winter (CSW) in \cite{CSW14, CSW10} when the Kochen-Specker rules of exclusivity and completeness are enforced. In formulating the graph-theoretic approach to quantum correlations, CSW had considered general non-contextuality inequalities $S$ as a positive linear combination of probabilities of events $S = \sum_i w_i P(e_i)$ with $w_i > 0$. For instance, the well-known KCBS inequality corresponding to the $5$-cycle exclusivity graph is of the form $S_{KCBS} = \sum_{i=0}^4 P(0,1|i,i+1) \leq 2$ with $2$ denoting the maximal value in all non-contextual hidden variable theories. In the CSW framework, one associates to every such non-contextuality inequality $S$ a vertex-weighted graph $(G, w)$ (note that a vertex-weighted graph $(G,w)$ is a graph $G$ with vertex set $V$ and weight assignment $w : V \rightarrow \mathbb{R}_+$).
The events $e_i$ appearing in $S$ are represented by vertices in $G$, adjacent vertices in $G$ represent exclusive events (events $e_i$ and $e_j$ are exclusive if there exist jointly measurable observables $\mu_i$ and $\mu_j$ that distinguish between the events), and the vertex weights represent the coefficients $w_i$ of the probabilities $P(e_i)$. The graph $(G,w)$ is then called the exclusivity graph of $S$. The main result of CSW is the following theorem showing how the exclusivity graph of $S$ can be used to calculate the optimal value of the inequality in classical and quantum theories.  
\begin{theorem}[Result $1$ of CSW \cite{CSW14}]
Given $S$ corresponding to a non-contextuality inequality, the maximum value of $S$ in classical and quantum theories is given by

\begin{equation}
    S \stackrel{\text{NCHV}}{\leq} \alpha(G, w) \stackrel{\text{Q}}{\leq} \theta(G,w),
\end{equation}
where $\alpha(G,w)$ is the independence number of $(G,w)$ and $\theta(G,w)$ is the Lov\'{a}sz-theta number of $(G,w)$. 
\end{theorem}

While it was recognized that $\theta(G,w)$ may only provide an upper bound to the quantum value in some cases - specifically when further constraints are imposed coming from the physical settings of the experiment, the statement that the classical value of any given non-contextuality inequality is given by $\alpha(G,w)$ is ubiquitous and taken to be true without any qualifications in much of the literature to this point. Furthermore, it is often also applied in situations when the Kochen-Specker rule of Completeness is enforced, namely that every maximum clique has exactly one vector that is assigned value $1$. We now make the observation that this statement needs to be carefully considered in computing the classical value of non-contextuality inequalities arising from gadget-type strucutres, $\alpha(G,w)$ only provides an upper bound to the classical value for such inequalities when Completeness is enforced. Specifically consider a non-contextuality inequality $S$ arising from a gadget-type exclusivity graph, so that vertex-weighted graph $(G,w)$ corresponding to $S$ is a gadget-type graph, meaning that the vertex set $V$ of $G$ can be partitioned into an independent set of distinguished vertices $V_{\text{dist}} \subset V$ and the non-distinguished vertices $V \setminus V_{\text{dist}}$. And furthermore, the weights $w_i$ in $S$ are given as

\begin{equation} 
\label{eq:weight-CSW}
w_i = \left\{ \begin{array}{ll} 
      1 & i \in V_{\text{dist}} \\
      0 & i \in V \setminus V_{\text{dist}} \\
   \end{array}
\right. 
\end{equation}

\begin{observation}
Given $S$ corresponding to a non-contextuality inequality, the maximum value of $S$ under the KS rules of Exclusivity and Completeness, in any classical (non-contextual hidden variable) theory is given by $\alpha(G,w)$ if and only if the vertex-weighted graph $(G,w)$ is not of gadget-type.
\end{observation}
\begin{proof}
The proof of the observation is a direct consequence of the gadget property of the exclusivity graph $(G,w)$. By this gadget property, it holds that the vertices in the distinguished set $V_{\text{dist}}$ cannot all be assigned value $1$ in any non-contextual assignment when Completeness is enforced, despite the fact that $V_{\text{dist}}$ is an independent set. Therefore, for the non-contextuality inequality $S = \sum_i w_i P(e_i)$ with $w_i$ given by Eq. (\ref{eq:weight-CSW}), the maximum value in any non-contextual hidden variable theory is only upper bounded by and never equal to $\alpha(G,w)$ (note that achieving $\alpha(G,w)$ requires assigning value $1$ to every vertex in the independent set $V_{\text{dist}}$). Furthermore, this restriction on the assignment of $1$s (arising from the KS requirement that every measurement returns an outcome) is exactly the defining feature of gadget-type graphs so that $\alpha(G,w)$ is not equal to the classical value (under Completeness) only for the non-contextuality inequalities arising from such measurement structures.
\end{proof}
\subsection{A generalisation of gadgets with other forbidden value assignments}
Thus far, we have generalised the well-known `True-implies-False' sets or $01$-gadgets to gadgets of order $(m,k)$. These latter sets contain $m$ independent vertices (mutually non-orthogonal vectors) of which exactly $k$ may be assigned value $1$ in any $\{0,1\}$-coloring. In other words, given an independent set $I = \{v_1, \dots, v_m\}$ these consider forbidden value assignments of the form $\{(\underbrace{1,\dots,1}_{\tiny{k+1}}, f(v_{k+2}),\dots, f(v_m)\}$ and permutations thereof, where $f: V \rightarrow \{0,1\}$ is any $\{0,1\}$-coloring of the vertices. One may consider a yet more general measurement structure in which we specify a set of forbidden value assignments $\mathcal{H} \subset \{0,1\}^m$. The usual `True-implies-False' sets then correspond to $\mathcal{H} = \{(1,1)\}$ for a given independent set $I = \{v_1, v_2\}$ of two vertices.

\begin{definition}
A $\{0,1\}$-colorable set $S_{gad} \in \mathbb{C}^d$ is a gadget in dimension $d$ for a specified forbidden set $\mathcal{H} \subset \{0,1\}^m$ if it contains an (ordered) set of mutually non-orthogonal vectors $I = \{|v_1 \rangle, \dots, |v_m \rangle\}$ such that in any $\{0,1\}$-coloring of $S_{gad}$ it holds that $f(I) \in \{0,1\}^m \setminus \mathcal{H}$. 
\end{definition}

Equivalently, the generalised gadget for a forbidden assignment $\mathcal{H}$ can be defined in graph-theoretic terms as follows. 

\begin{definition} 
A $\{0,1\}$-colorable graph $G = (V,E)$ is a gadget in dimension $d$ for a specified forbidden set $\mathcal{H} \subset \{0,1\}^m$, if it has faithful
dimension $d^*(G_{gad})=\omega(G_{gad})=d$ and contains an (ordered) independent set $I \subset V$ of cardinality $|I| = m$ such that in any $\{0,1\}$-coloring $f: V \rightarrow \{0,1\}$ it holds that $f(I) \in \{0,1\}^m \setminus \mathcal{H}$. 

%For a specified forbidden-value assignment $\mathcal{H} \subset \{0,1\}^m$, a $\{0,1\}$-colorable graph $G = (V,E)$ is a gadget-for-$\mathcal{H}$.

%$\mathcal{H}\subsetneq \{0,1\}^{\times k}$, a general gadget $G=(V_G,E_G)$ is a graph that contains an independent set $I\subset V_G$ of size $|I|=k$ such that for any valid $\{0,1\}$-assignments $f: V_G\to\{0,1\}$, it is the case that $f(I)\in\mathcal{H}$.
\end{definition}

A natural question then arises - can a gadget be constructed for every forbidden set $\mathcal{H} \subset \{0,1\}^m$ for arbitrary $m \geq 2$? Note that the number of such forbidden sets $\mathcal{H}$ is $2^{2^m} - 1$ (i.e., all subsets of $\{0,1\}^m$ except the empty set). We answer this question in the affirmative and demonstrate the precise steps of building such a gadget in \revise{Supplementary Information Note 4}.  It's worth noting that Kochen Specker proofs themselves come under the umbrella of the generalised gadget structures defined here, with $\mathcal{H} = \{0,1\}^m$ for some independent set in the graph $I$ of size $|I| = m$ and for arbitrary $m \geq 1$. In other words, there is no valid $\{0,1\}$-assignment to the vertices of the independent set $I$.

 \section*{Data Availability}
 The authors declare that the data supporting the findings of this study are available within the paper and in the Supplementary Information.

 \section*{Acknowledgements}
\revise{L.Y. and R.R. acknowledge support from the Start-up Fund ”Device-Independent Quantum Communication Networks” from The University of Hong Kong, the Early Career Scheme (ECS) grant "Device-Independent Random Number Generation and Quantum Key Distribution with Weak Random Seeds" (Grant No. 27210620), the General Research Fund (GRF) grant "Semi-device-independent cryptographic applications of a single trusted quantum system" (Grant No. 17211122) and the Research Impact Fund (RIF) "Trustworthy quantum gadgets for secure online communication" (Grant No. R7035-21).}
K.H., M.R. and P.H. acknowledge partial support by the Foundation for Polish Science (IRAP project, ICTQT, contract no. MAB/2018/5, co-financed by EU within Smart Growth Operational Programme). The `International Centre for Theory of Quantum Technologies' project (contract no. MAB/2018/5) is carried out within the International Research Agendas Programme of the Foundation for Polish Science co-financed by the European Union from the funds of the Smart Growth Operational Programme, axis IV: Increasing the research potential (Measure 4.3).

\section*{Competing Interests}
The authors declare no competing interests.

\section*{Author Contributions}
The project was conceived by R.R., L.Y., and K.H., and all authors discussed extensively and contributed to the development of the theory and content presented in the manuscript. The initial draft of the manuscript was written by R.R., L.Y., and K.H., and it was critically reviewed and revised by P.H. before the final approval of the completed version by all authors.

\clearpage

\section*{Supplementary Information for `` Optimal Measurement Structures for Contextuality Applications''}

\subsection* {Supplementary Note 1: The construction of an order $(d,d-1)$ gadget for general dimension $d$ }\label{speical_gagdet}
We first show the order $(3,2)$ gadget with the three distinguished vectors $|m_1 \rangle, |m_2 \rangle$ and $|m_{3} \rangle$ being arbitrary close to each other as a special case, and then generalize this construction to order $(d,d-1)$ cases for arbitrary $d$.
\begin{figure}[htbp]
    \centering
    \includegraphics[width=18cm]{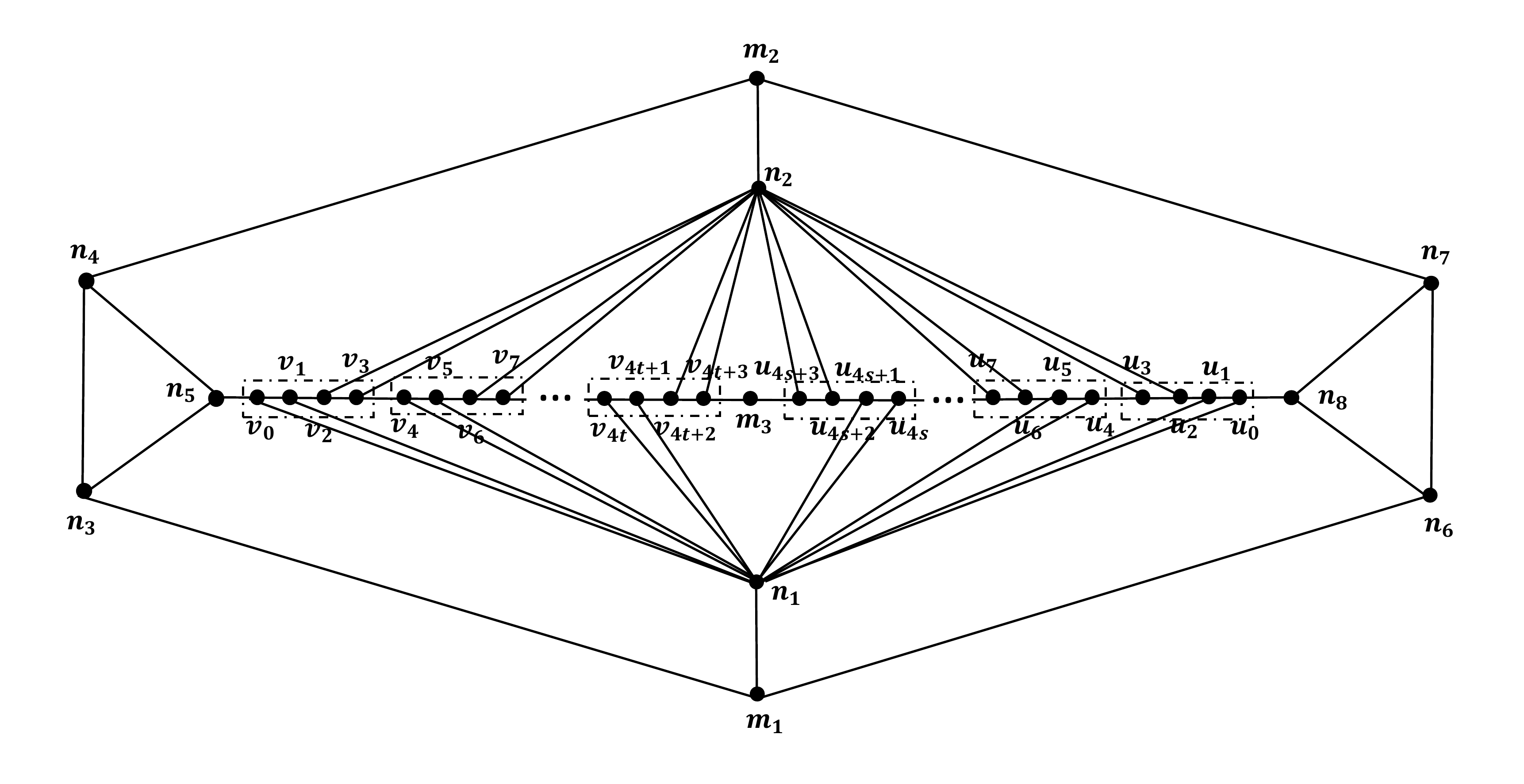}
    \caption{A construction of an order $(3,2)$ gadget. As the number of repeating unit of $4$ vectors (denoted by dashed rectangles) $t,s\to\infty$, the construction serves as a gadget for distinguished vectors $|m_1 \rangle, |m_2 \rangle$ and $|m_3 \rangle$ (with arbitrary overlaps).}
    \label{special_3}
\end{figure}

The particular set of vectors that give the orthogonal representation of the order $(3,2)$ gadget in Fig.\ref{special_3} is as follows:

\begin{equation}
    \begin{split}
    |n_1\rangle&=(	0,  1,  0)^T, \qquad 
    |n_2\rangle=(	-\cos\theta,  \sin\theta,  0	)^T\\
    |n_3\rangle&=(	0,  \sin\theta,  -\cos\theta)^T, \qquad
    |n_4\rangle=(	-\cos^2\theta, (1/2) \sin{2\theta}, \sin^2\theta)^T\\
    |n_5\rangle&=(	\sin\theta,  \cos^3\theta,  \sin\theta\cos^2\theta)^T,\qquad
    |n_6\rangle=(0,  \sin\phi,  -\cos\phi)^T\\
    |n_7\rangle&=(	-\cos\theta\cos\phi,  \sin\theta\cos\phi,  \sin\theta\sin\phi)^T, \qquad
    |n_8\rangle=(\sin\theta,  \cos\theta\cos^2\phi,  (1/2)\sin{2\phi}\cos\theta)^T\\
    \end{split}
\end{equation}

\begin{equation}
    \begin{split}
    &|v_{4t}\rangle=(	-\cos^2\theta,  0, \sin^{2t}\theta)^T, \qquad
    |v_{4t+1}\rangle=(	\sin^{2t}\theta,  0, \cos^2\theta)^T\\
    &|v_{4t+2}\rangle=(	\sin\theta\cos^2\theta,  \cos^3\theta,  -\sin^{2t+1}\theta)^T, \qquad
    |v_{4t+3}\rangle=(	\sin^{2t+2}\theta,  \sin^{2t+1}\theta\cos\theta, \cos^2\theta)^T\\
    &|u_{4s}\rangle=(-(1/2)\sin{2\phi}\cos\theta,  0,  \sin^{2s+1}\theta)^T, \qquad
    |u_{4s+1}\rangle=(\sin^{2s+1}\theta,  0, (1/2)\sin{2\phi}\cos\theta
    )^T\\
    &|u_{4s+2}\rangle=(	(1/4)\sin{2\theta}\sin{2\phi},  (1/2)\sin{2\phi}\cos^2\theta,  -\sin^{2s+2}\theta)^T, \qquad
    |u_{4s+3}\rangle=(\sin^{2s+3}\theta, \sin^{2s+2}\theta\cos\theta,(1/2)\sin{2\phi}\cos\theta)^T\\    
    \end{split}
\end{equation}
with the distinguished vectors given by:

\begin{equation}
    \begin{split}
    |m_1\rangle& =(1,0,0)^T, \qquad
    |m_2\rangle=(\sin\theta, \cos\theta, 0)^T, \qquad 
    |m_3\rangle=\left(	\sin\theta, \cos\theta,-\frac{\sin^{2t+1}\theta}{\cos^2\theta}\right)^T.\\
    \end{split}
\end{equation}
Here $\theta,\phi\in(0,\pi/2)\cup (\pi/2,\pi)$ are parameters to be fixed, and $t, s \geq 1$ denote the number of repeating units from the left (from vertex $n_5$) and the right (from vertex $n_8$) respectively. 

As $t\to\infty,s\to\infty$, in the middle of the figure, we obtain

\begin{equation}
    \begin{split}
        \langle v_{4t+3}|m_3\rangle=&0, \qquad \; \text{and} \; \qquad 
        \langle u_{4s+3}|m_3\rangle=\sin^{2s+4}\theta+\sin^{2s+2}\theta\cos^2\theta
        -\frac{\sin^{2t+1}\theta\sin\phi\cos\theta\cos\phi }{\cos^2\theta}=0\\
    \end{split}
\end{equation}
satisfying the orthogonality constraints of the graph, and when in addition $\theta\to \frac{\pi}{2}$, we obtain that $\langle m_i|m_j\rangle\to1$ for $i \neq j$ showing that the construction serves as a gadget for vectors that are arbitrarily close to each other. 

For dimension $d\geq 4$, one can construct an order $(d,d-1)$ gadget for arbitrary distinguished vectors as shown in the graph in Figure \ref{special_d}. 
\begin{figure}
    \centering
    \includegraphics[width=16cm]{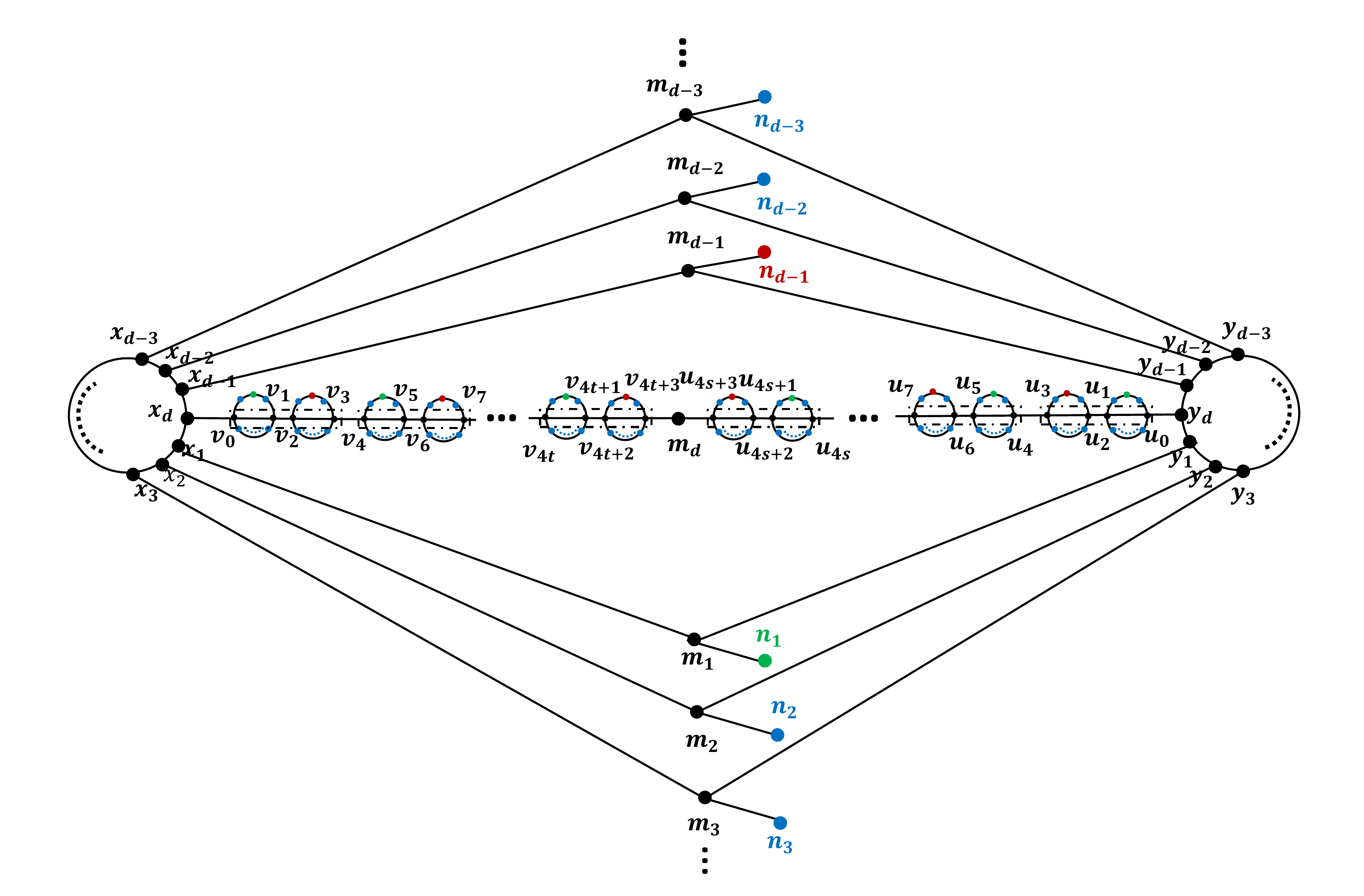}
    \caption{The construction of an order $(d,d-1)$ gadget. As the number of repeating units $t, s \to \infty$, the construction works for arbitrary distinguished vectors $|m_1 \rangle, \dots, |m_d \rangle$, i.e., $\langle m_i | m_j \rangle \to 1$.}
    \label{special_d}
\end{figure}

A faithful orthogonal representation of the graph is given by the following list of $d$-dimensional $(d\geq5)$ vectors: 

\begin{equation}
    \begin{split}
\langle m_1| &=(	1,  0,  0, 0,0, \cdots,0,  0,  0), \qquad \qquad \qquad   
\langle m_2| =(
	\sin\theta, 0,  0,0,0,  \cdots, 0,\cos\theta,  0), \\
\langle m_3| &=(	\sin\theta,  0,  0,0, 0, \cdots,\cos\theta,0,  0), \qquad
\langle m_{d-3}| =(	\sin\theta,  0,  0,\cos\theta,0,  \cdots,0,0,  0),\\
\langle m_{d-2}| &=(	\sin\theta,  0, \cos\theta, 0,0,  \cdots,0,0,  0), \qquad
\langle m_{d-1}| =(	\sin\theta,  \cos\theta,  0,0,0, \cdots,0,0,  0),\\
\langle n_1| &=(	0,  1,  0, 0,0,\cdots,0,  0,  0), \qquad \qquad \qquad
\langle n_2 | =(	0,  0,  1,0, 0, \cdots, 0,0,  0),\\
\langle n_3| &=(	0,  0,  0,0, 0, \cdots,0,1,  0), \qquad \qquad \qquad
\langle n_4| =(	0,  0,  0,0, 0, \cdots,1,0,  0),\\
\langle n_{d-3}| &=(	0,  0, 0,0,1, \cdots,0,0,  0), \qquad \qquad \quad
\langle n_{d-2}| =(	0,  0, 0, 1, 0, \cdots,0,0, 0),\\
\langle n_{d-1}| &=(	-\cos\theta,  \sin\theta,  0,0,0, \cdots,0,0,  0) \qquad  \quad
\langle x_1| =(	0,  b,  b, b,b, \cdots,b,  b,  c),\\
\langle x_2| &=(	-\cos\theta, a,  a,a, a, \cdots, a,\sin\theta,  \cos\theta), \qquad \quad
\langle x_3| =(	-\cos\theta,  a,  a,a, a, \cdots,\sin\theta,a,  \cos\theta),\\
\langle x_{d-3}| &=(	-\cos\theta,  a,  a,\sin\theta,a, \cdots,a,a,  \cos\theta), \qquad \quad
\langle x_{d-2}| =(	-\cos\theta,  a, \sin\theta, a, a, \cdots,a,a,  \cos\theta),\\
\langle x_{d-1}| &=(	-\cos\theta,  \sin\theta,  a,a,a, \cdots,a,a,  \cos\theta), \qquad \quad
\langle x_d| =(	p,  e,  e, e,e, \cdots,e,  e, q),\\
\langle y_1| &=(	0,  b',  b', b',b', \cdots,b',  b',  c'), \qquad \qquad \qquad\quad
\langle y_2| =(	-\cos\theta,  a',  a',a', a', \cdots, a',\sin\theta,  \cos\phi),\\
\langle y_3| &=(	-\cos\theta,  a',  a',a', a', \cdots,\sin\theta,a',  \cos\phi), \qquad \quad
\langle y_{d-3}| =(	-\cos\theta,  a',  a',\sin\theta,a', \cdots,a',a',  \cos\phi),\\
\langle y_{d-2}| &=(	-\cos\theta,  a',\sin\theta, a', a', \cdots,a',a',  \cos\phi), \qquad \quad
\langle y_{d-1}| =(	-\cos\theta,  \sin\theta, a',a',a', \cdots,a',a',  \cos\phi),\\
\langle y_d| &=(	p',  e',  e', e',e', \cdots,e',  e',  q'), \qquad \qquad \qquad \qquad
\langle v_{4t}| =(	-q, 0, 0,0,0, \cdots,0,0,  p\sin^{2t}\theta),\\
\langle v_{4t+1}| &=(	p\sin^{2t}\theta,  0,  0,0,0, \cdots,0,0,  q), \qquad \qquad \qquad \qquad 
\langle v_{4t+2}| =(	q\sin\theta,  q\cos\theta,  0,0,0, \cdots,0,0,  -p\sin^{2t+1}\theta),\\
\langle v_{4t+3}| &=(	p\sin^{2t+2}\theta,  p\sin^{2t+1}\theta\cos\theta,  0,0,0, \cdots,0,0, q), \qquad  
\langle u_{4s}| =(	-q',  0,  0,0,0, \cdots,0,0,  p'\sin^{2s}\theta),\\
\langle u_{4s+1}| &=(	p'\sin^{2s}\theta,  0,  0,0,0, \cdots,0,0,  q'),\qquad \qquad \qquad \qquad
\langle u_{4s+2}|=(	q'\sin\theta,  q'\cos\theta,  0,0,0, \cdots,0,0, -p'\sin^{2s+1}\theta),\\
\langle u_{4s+3}| &=(	p'\sin^{2s+2}\theta,  p'\sin^{2s+1}\theta\cos\theta,  0,0,0, \cdots,0,0,  q'),\qquad 
\langle m_d |=(	\sin\theta, \cos\theta,  0,0,0, \cdots,0,0, -\frac{p}{q}\sin^{2t+1}\theta),\\
    \end{split}
\end{equation}
where $\theta,\phi\in(0,\pi/2)\cup (\pi/2,\pi)$, and $t, s$ denote the number of repeating units from the vertex $x_d$ and $y_d$ respectively.

In order to satisfy the orthogonality constraints of the graph, we set  $a=\frac{-\sin\theta+\sqrt\Delta}{d-4};b=\frac{d-4}{d-3}\cos\theta, c=\frac{1}{d-3}\sin\theta-\sqrt\Delta;e=\frac{d-3}{(d-2)(d-4)}\cos\theta ;q=\frac{\cos^2\theta}{\frac{-1}{d-3}\sin\theta+\sqrt\Delta};p=\frac{d-3}{(d-2)(d-4)}\sin\theta+\frac{\cos^2\theta}{\frac{-1}{d-3}\sin\theta+\sqrt\Delta}+\frac{(d-3)^2}{(d-2)(d-4)^2}(-\sin\theta+\sqrt\Delta)$ and $a'=\frac{-\sin\theta+\sqrt\delta}{d-4};b'=\frac{d-4}{d-3}\cos\phi, c'=\frac{1}{d-3}\sin\theta-\sqrt\delta;e'=\frac{d-3}{(d-2)(d-4)}\cos\phi ; q'=\frac{\cos^2\phi}{\frac{-1}{d-3}\sin\theta+\sqrt\delta};p'=\frac{d-3}{(d-2)(d-4)}\sin\theta\frac{\cos\phi}{\cos\theta}+\frac{\cos^2\phi}{\frac{-1}{d-3}\sin\theta+\sqrt\delta}\frac{\cos\phi}{\cos\theta}+\frac{(d-3)^2}{(d-2)(d-4)^2}(-\sin\theta+\sqrt\delta)\frac{\cos\phi}{\cos\theta}$, 
where $\Delta=\sin^2\theta-2(d-4)\cos^2\theta\geq 0;
\delta=\sin^2\theta-(d-4)(\cos^2\theta+\cos^2\phi)\geq 0$. Remark that the vectors $|n_2\rangle,|n_3\rangle, \cdots ,|n_{d-2}\rangle$ appear in every one of the repeating units and thus the blue vertices in the figures must be identified with each other, while $|n_1\rangle$ (the green vertices) and $|n_{d-1}\rangle$ (the red vertices) appear in alternate repeating units. 

As $t, s\to\infty$, we obtain that

\begin{equation}
    \begin{split}
      \langle v_{4t+3}|m_d\rangle&=0\\
      \langle u_{4s+3}|m_d\rangle&=p'\sin^{2s+3}\theta+p'\sin^{2s+1}\theta\cos^2\theta-\frac{pq' }{q}\sin^{2t+1}\theta=0.
    \end{split}
\end{equation}
And if furthermore $\theta\to\frac{\pi}{2}$, we obtain that

\begin{equation}
    \langle m_i|m_j\rangle_{i\neq j}\to1,
\end{equation}
that is, the distinguished vectors tend to converge towards each other. 

We omit the orthogonal representation for $d=4$ here. The vectors have a similar form, but the coefficient needs to be slightly adjusted.

%\end{itemize}
\subsection*{Supplementary Note 2: Using higher-order gadgets to construct Kochen-Specker proofs and state-independent contextuality proofs}\label{App:SIC-const} 
The construction in dimension $d$ of order $(k,k-1)$ gadgets with $2 \leq k \leq d$ from the previous section can be used to efficiently build Kochen-Specker sets as well as state-independent contextual (SI-C) sets of the Yu-Oh type. 

Fix a value of $k$ in the range $\{2,\dots, d\}$. 

\begin{construction}\label{cos_KS}
The gadgets of order $(k,k-1)$ can be used as building blocks (together with a set of bases) to construct Kochen Specker proofs in dimension $d$.
\end{construction}
\begin{itemize}
    \item [Step $1$]
    We begin with $k$ bases sets in dimension $d$, denoted as $B_1, B_2, \dots, B_k$. We choose these sets such that no two vectors in different bases sets are identical or orthogonal to each other (one can do this by picking a single basis set $B_1$ and applying a suitable unitary matrix $U_d$ to $B_1$). 
    \item [Step $2$]
    Construct all possible sets $S_i = \big\{ |v_{B_p}^q \rangle \big\}$ with $p \in [k] := \{1,\dots,k\}$ and $q \in [d]$, obtained by choosing a single vector $|v_{B_p}^q \rangle$ from each basis set $B_p$. In total, we thus have $d^k$ sets $S_i$ with $|S_i| = k$ for each $i \in [d^k]$.
    \item [Step $3$]
    Construct for each $i \in [d^k]$ an order $(k,k-1)$ gadget in dimension $d$ with the vectors in the set $S_i$ being the distinguished vectors. Such a gadget can be built following the construction in the previous section, notice that an order $(k,k-1)$ gadget in dimension $k$ serves also as an order $(k,k-1)$ gadget in all dimensions $d \geq k$ by the addition of computational basis vectors $|k+1\rangle, \dots, |d\rangle$. 
    
    \end{itemize}

    All the vectors in $B_1\cup B_2\cup \cdots \cup B_k \cup\mathcal{S}$ form a KS proof, where $\mathcal{S}$ denotes all the high-order gadgets used in Step $3$. This follows from the fact that assigning a single value $1$ to each of the bases $B_1, \dots, B_{k-1}$ forces all the vectors in the basis set $B_k$ to be assigned value $0$ giving rise to a contradiction. 
\qed

\begin{construction}
The gadgets of order $(k,k-1)$ can be used as building blocks (together with a set of bases) to construct general SI-C sets in dimension $d$. 

\end{construction}

To realize the construction we need the following Theorem.

\begin{theorem}\label{state_ind}
Set 
%\begin{equation}
$n = \begin{cases} 
      \lceil\log_2 \frac{d-1}{2}\rceil, & $d$ \; \text{is odd} \\
      \lceil\log_2 \frac{d-2}{2}\rceil, & $d$ \; \text{is even}
   \end{cases}.$
%\end{equation}
Let $r \geq 4$ be an even integer. There exist $r \cdot 2^n$ distinct unit vectors $|u_i \rangle$ in dimension $d$ satisfying

\begin{equation}\label{SIC}
\sum_{i=1}^{r\cdot 2^n}|u_{i}\rangle\langle u_{i}|=\frac{r\cdot 2^n}{d} \mathbb{1}_d
\end{equation}
\end{theorem}
\begin{proof}

Firstly, we show the construction of these $r \cdot 2^n$ distinct unit vectors $|u_i \rangle$. Then we prove the above statement. We just consider the case that $d$ is an odd integer in the following construction and briefly explain the case that $d$ is even at the end of the proof. 

\textit{Construction.}
To give the explicit form of these $r \cdot 2^n$ distinct unit vectors $|u_i \rangle$, we consider a matrix $O$ of size $d\times r\cdot 2^{n}$ whose $i$-th column is $|u_i\rangle$. Thus, we have 

\begin{equation}
\sum_{i=1}^{r\cdot 2^{n}}|u_{i}\rangle\langle u_{i}|=O\cdot O^T
\end{equation}

\begin{itemize}
\item [Step $1$]
Divide the matrix $O$ into blocks by each $r$ columns and term every block matrix $o_i$. So the size of each block matrix $o_i$ is $d\times r$, and there are $2^n$ block matrices in total, i.e., $i\in[2^{n}] $.
The $j$-th ($j\in[r]$) column in the block matrix $o_i$ is $|u_{(i-1)\cdot r+j}\rangle$. Construct it as:

\begin{equation}\label{SIC-vector}
    (\underbrace{q\cos(\frac{j}{r}\cdot2\pi),q\cos(\frac{j}{r}\cdot2\pi),\cdots,q\cos(\frac{j}{r}\cdot2\pi)}_{\frac{d-1}{2} },\underbrace{q\sin(\frac{j}{r}\cdot2\pi),q\sin(\frac{j}{r}\cdot2\pi),\cdots,q\sin(\frac{j}{r}\cdot2\pi)}_{\frac{d-1}{2}},(-1)^j\cdot p)^{T}
\end{equation}
To satisfy the normalization condition, set $p:=\sqrt{1-\frac{q^2(d-1)}{2}}$.
We set the value of $q$ later.
\item [Step 2]
Set $n:=\lceil\log_2 \frac{d-1}{2}\rceil$, there exist a Hadamard matrix $H_{{2^n\times2^n}}$ of size $2^n\times2^n$. We take arbitrary  $\frac{d-1}{2}$ rows from this Hadamard matrix $H_{2^n\times2^n}$ and arrange them in rows (in arbitrary order) to form a new matrix $H'_{\frac{d-1}{2}\times2^n}$ of size $\frac{d-1}{2}\times2^n$. 

 \item [Step 3]
 For each vector $|u_{(i-1)\cdot r+j}\rangle,j\in[r],i\in[2^n]$ in $i$-th block matrix $o_i,i\in[2^n]$ (i.e., the $j$-th ($j\in[r]$) column in each block matrix  $o_i,i\in[2^n]$), take the sign of the $i$-th column of $H'$ as a sign vector $s_i$ (size $1\times \frac{d-1}{2}$). Set the sign $s_i$ to the first $\frac{d-1}{2}$ entries ( i.e., the entries that contain $\cos$ function in \eqref{SIC-vector} ) and the second $\frac{d-1}{2}$ entries ( i.e., the entries that contain $\sin$ function in \eqref{SIC-vector}) of each vector  $|u_{(i-1)\cdot r+j}\rangle,j\in[r],i\in[2^n]$ in $i$-th block matrix $o_i,i\in[2^n]$. Keep the sign of the last entry ( i.e., $(-1)^j\cdot p$).
\end{itemize}

With the above construction of these $r \cdot 2^n$ distinct unit vectors $|u_i \rangle$, we can now prove the main statement \eqref{SIC} of this theorem by showing $O\cdot O^T$ is a diagonal matrix, and diagonal elements are identical.

Since $O\cdot O^\mathsf{T}$ is symmetric, we calculate the following cases:

\begin{itemize}
    \item [1.]
    $\forall i<j \in[\frac{d-1}{2}]$:
    
    \begin{equation}
        (O\cdot O^{T})_{i,j}=q^2\left(\sum_{k=1}^{r}\cos^2(\frac{k}{r}\cdot 2\pi)\right)\cdot(H'_{i,\_}\cdot {H'_{j,\_}} ^{T})=0
    \end{equation}
    where $H'_{i,\_}$ is the $i$-th row of matrix $H'_{\frac{d-1}{2}\times2^n}$.
     \item [2.]
    $\forall i<j \in\{\frac{d-1}{2}+1,\cdots,d-1\}$:
    
    \begin{equation}
        (O\cdot O^{T})_{i,j}=q^2\left(\sum_{k=1}^{r}\sin^2(\frac{k}{r}\cdot 2\pi)\right)\cdot(H'_{i-\frac{d-1}{2},\_}\cdot {H'_{j-\frac{d-1}{2},\_}}^{T})=0
    \end{equation}
     \item [3.]
    $\forall i \in[\frac{d-1}{2}]$ and $j\in\{\frac{d-1}{2}+1,\cdots,d-1\}$:
    
    \begin{equation}
        (O\cdot O^{T})_{i,j}=q^2\left(\sum_{k=1}^{r}\cos(\frac{k}{r}\cdot 2\pi)\sin(\frac{k}{r}\cdot 2\pi)\right)\cdot(H'_{i,\_}\cdot {H'_{j-\frac{d-1}{2},\_}}^{T})=0
    \end{equation}
    \item [4.] $i\in[\frac{d-1}{2}]$ and $j=n$:
    
    \begin{equation}
        (O\cdot O^{T})_{i,j}=pq\left(\sum_{k=1}^{r}\cos(\frac{k}{r}\cdot 2\pi)\right)\cdot\sum_{k=1}^{2^n} H'_{i,k}=0
    \end{equation}
    \item [5.] $i\in\{\frac{d-1}{2}+1,\cdots,d-1\}$ and $j=n$:
    
    \begin{equation}
        (O\cdot O^{T})_{i,j}=pq\left(\sum_{k=1}^{r}\sin(\frac{k}{r}\cdot 2\pi)\right)\cdot\sum_{k=1}^{2^n} H'_{i-\frac{d-1}{2},k}=0
    \end{equation}
    \item [6.] $i=j \in[\frac{d-1}{2}]$:
    
        \begin{equation}
        \begin{split}
        (O\cdot O^{T})_{i,j}&=q^2\left(\sum_{k=1}^{r}\cos^2(\frac{k}{r}\cdot 2\pi)\right)\cdot 2^n\\
        &=q^2\left(\sum_{k=1}^{r}\frac{1+\cos(\frac{2k}{r}\cdot 2\pi)}{2}\right)\cdot 2^n\\
        &=q^2\cdot \frac{r}{2}\cdot2^{n}+0=q^2 r\cdot2^{n-1}
        \end{split}
    \end{equation}
    \item [7.] $i=j \in\{\frac{d-1}{2}+1,\cdots,d-1\}$:
    
        \begin{equation}
        \begin{split}
        (O\cdot O^{T})_{i,j}&=q^2\left(\sum_{k=1}^{r}\sin^2(\frac{k}{r}\cdot 2\pi)\right)\cdot 2^n\\
        &=q^2\left(\sum_{k=1}^{r}\frac{1-\cos(\frac{2k}{r}\cdot 2\pi)}{2}\right)\cdot 2^n\\
        &=q^2\cdot \frac{r}{2}\cdot2^{n}+0=q^2 r\cdot2^{n-1}
        \end{split}
    \end{equation}
    \item [8.] $i=j=n$
    
    \begin{equation}
        (O\cdot O^{T})_{i,j}=p^2r\cdot2^{n}=(1-\frac{q^2(d-1)}{2})r\cdot2^{n}
    \end{equation}
\end{itemize}
To make all diagonal terms identical, we have:

\begin{equation}
    q^2 r\cdot2^{n-1}=(1-\frac{q^2(d-1)}{2})r\cdot2^{n} \Rightarrow q=\sqrt{\frac{2}{d}}
\end{equation}
So the diagonal terms are $\frac{r\cdot2^n}{d}$. 
In summary, $\sum_{i=1}^{r\cdot 2^{n}}|u_{i}\rangle\langle u_{i}|=O\cdot O^{T}=\frac{r\cdot2^n}{d}\mathbb{1}_d$. 

The construction of vectors and proof of theorem for the case that $d$ is an even integer is similar to the above case that $d$ is odd. The differences are: 
1. At the first step of the construction, the $j$-th ($j\in[r]$) column $|u_{(i-1)\cdot r+j}\rangle$ in block matrix $o_i$ is constructed as:

    \begin{equation}
 (\underbrace{q\cos(\frac{j}{r}\cdot2\pi),q\cos(\frac{j}{r}\cdot2\pi),\cdots,q\cos(\frac{j}{r}\cdot2\pi)}_{\frac{d-2}{2} },\underbrace{q\sin(\frac{j}{r}\cdot2\pi),q\sin(\frac{j}{r}\cdot2\pi),\cdots,q\sin(\frac{j}{r}\cdot2\pi)}_{\frac{d-2}{2}},1\cdot p,(-1)^j\cdot p)^{T}
\end{equation}
To satisfy the normalization condition, set $p:=\sqrt{\frac{1}{2}-\frac{q^2(d-2)}{4}}$. And the value of $q$ would be the same as above $q=\sqrt{\frac{2}{d}}$.

2. At the second step of the construction, set $n:=\lceil\log_2 \frac{d-2}{2}\rceil$ and take arbitrary  $\frac{d-2}{2}$ rows from this Hadamard matrix $H_{2^n\times2^n}$ to form the new matrix $H'_{\frac{d-2}{2}\times2^n}$ of size $\frac{d-2}{2}\times2^n$.  So the size of the sign vector $s_i,i\in[2^n]$ in step 3 should be $1\times \frac{d-2}{2}$.

\end{proof}

Fix a value of $k$ in the range $\{2,\dots, d\}$. We now use the above theorem along with order $(k,k-1)$ gadgets from the previous section to show a construction of SI-C sets.

\begin{itemize}
    \item [Step $1$] Choose an even integer $r> \max \left\{\frac{d(k-1)}{2^n}, 4 \right\}$ and construct the set of $r \cdot 2^n$ unit vectors $|u_i \rangle$ from Theorem \ref{state_ind}.
    \item [Step $2$]
    Construct all $\binom{r\cdot 2^{n}}{k}$ possible sets $S_i$ of size $k$ built out of the $r \cdot 2^n$ unit vectors $|u_i \rangle$. 
    \item [Step $3$]
    Construct sets $S'_i$ from $S_i$ by deleting all the mutually orthogonal vectors in the set $S_i$ so that $|S'_i| \leq k$. Construct for each $i$, an order $(|S'_i|, |S'_i|-1)$ gadget in dimension $d$ with the vectors in $S'_i$ constituting the distinguished vectors of the gadget. Such a gadget can be built following the construction in the previous section as before.

\end{itemize}

    The $r\cdot 2^{n}$ unit vectors $|u_i \rangle$ from Theorem \ref{state_ind} together with the set $\mathcal{S}$ of all higher-order gadgets used in Step $3$ constitute a state-independent contextuality proof. This follows from the fact that in any $\{0,1\}$-assignment $f$, it holds that
    
     \begin{equation}
     \sum_{i=1}^{r\cdot 2^{n}} f\left(\left|u_{i}\right\rangle\right) \leq k-1.
     \end{equation}
    On the other hand, in quantum theory, every state $|\psi \rangle \in \mathbb{C}^d$ achieves 
    
    \begin{equation}
     \sum_{i=1}^{r\cdot 2^{n}} \Tr\left[|u_{i}\left\rangle\right\langle u_{i}||\psi\left\rangle\right\langle \psi|\right]=\frac{r\cdot 2^n}{d}.
 \end{equation}
 We see that the quantum value is greater than the classical value for $r$ chosen as in Step $1$.
\qed
 \subsection*{Supplementary Note 3: Order $(k,k-1)$ gadgets as induced subgraphs in every Kochen-Specker set}\label{App}

\begin{theorem}
Every non-\{0,1\}-colorable graph $G$ contains a gadget of order $(k,k-1)$ for some $k$ satisfying $2\leq k \leq \omega(G)$.
\end{theorem}
\begin{proof}
In what follows we will use brute-force and inefficient greedy algorithm for $\{0,1\}$-coloring. A way to color any graph is to
\begin{enumerate}
\item identify the list of cliques of size $\omega(G)$, ${\cal C}=\{C_1,...,C_l\}$ in $G$.
\item for $i=1$ to $l$ do the following:
        \begin{enumerate}
        \item repeat until a vertex in $C_i$ is set to $1$ or all the vertices can not be set value $1$:
        Choose an as yet unvisited vertex $v\in C_i$ and check if it can be attributed $1$ at the current stage of execution. If not, mark the vertex as visited and go to (a). If yes, continue.
        
        \item if some vertex in $C_i$ has been assigned value $1$, and $i\neq l$ continue step $2$ with $i:=i+1$.
        \item if some vertex in $C_i$ has been assigned value $1$ and $i=l$ stop the algorithm and return SUCCESS.
        \item if all the vertices of $C_i$ are marked as visited and $i\neq 1$ unmark all of them and set $i:=i-1$. 
        \item if all the vertices of $C_i$ are marked as visited and $i=1$ stop the algorithm and return FAILURE.
        \end{enumerate}
\end{enumerate}
We first observe, that the graph is non-$\{0,1\}$-colorable if and only if the the condition of the step $(d)$ is satisfied at some point of its execution.

For the (only if) it is clear that if the condition of step $(d)$ is never satisfied during execution of the algorithm, it terminates with SUCCESS. That is to exactly one vertex of each of all the $l$ maximal cliques in $G$
it sets number $1$. In such a case setting all other vertices $0$ gives then a $\{0,1\}$-coloring of a graph, hence $G$ is $\{0,1\}$-colorable.

For (if) we will argue that the condition of the step $(e)$ is not satisfied (resulting in final FAILURE step) if the condition of the step $(d)$ is not satisfied beforehand. 
This is because the condition of the step $(e)$ is satisfied only when $i=1$ and all vertices of $C_1$ have been visited. It can not happen in the first execution of $i=1$ in step $2$ as any (of at least  $2$) of vertices of $C_1$, can be set $1$ with no conflict with the $\{0,1\}$-coloring rule. Hence the index $i$ in execution of this algorithm has to be at least equal to $2$ at some point. The only way then to satisfy later condition of the step $(e)$ is to lower $i$ by at least $1$ at some point of execution. And the only step in which $i$ can be lowered is the step $(d)$.
Hence the condition of the step $(d)$ must be satisfied at least once before the condition of the step $(e)$ is satisfied, that is if the graph $G$ is not $\{0,1\}$-colorable i.e. when algorithms has to terminate with FAILURE.

Let now $G$ be not $\{0,1\}$-colorable. The above
algorithm has to report FAILURE i.e. condition of the step $(e)$ and hence also of the step $(d)$ has to be satisfied at some point as argued above. To be concrete, consider now the lowest index $2\leq j\leq l$ such that the condition of the step $(d)$ is satisfied with index of the currently considered clique $C_i$ equal $i=j$.
Due to definition of 
step $(d)$, the clique $C_j$ can not be attributed $1$ according to $\{0,1\}$-coloring. The only way so that the $\{0,1\}$-coloring rules are not satisfied when any of $v\in V_{C_j}$ gets $1$, is that every vertex of the clique $C_j$ is adjacent to some vertex that have been set $1$ during previous steps of this execution of the algorithm. 

In other words, that there exists
set $D=\{w_1,...,w_s\}\subset V$, such that dominates the clique $C_j$ and these vertices got $1$ in previous steps. Note that $s=j-1$, as there were $j-1$ cliques which got some vertex set to $1$ with no contradiction with $\{0,1\}$-coloring until the first contradiction with $\{0,1\}$-coloring happened at $i=j$. Since $j$ is the lowest value of the index $1\leq i\leq l$ when the $\{0,1\}$-coloring rule is broken, in previous steps of this execution there were no contradiction to $\{0,1\}$-coloring. This means that the set $D$ is an independent set in $G$ (only non-adjacent vertices could have been given $1$ with no contradiction to $\{0,1\}$-coloring).

Consider then a subset $D'$ of $D$ which is minimal with the property that it dominates $C_j$ in $G$. Further consider an induced subgraph $H< G$ with the vertex set $V_H:=D'\cup V_{C_j}$ and edges from $G$. We claim now
that this set is a $01$-gadget of order $(k,k-1)$ for some $k$ in the range $2\leq k\leq \omega(G)$. 

We first note that the points $(i)$ and $(ii)$ from the definition of order $(m,k)$ gadget are satisfied for $H$. To see the latter note first that the vertices from $D'$ form independent set. However when
all of them are set to $1$ the clique $C_j$ which is dominated by these vertices can not have any $1$ which contradicts the $\{0,1\}$-coloring rules.
On the other hand, if any proper subset $D''\subsetneq D'$ is set $1$ some of the vertices of $C_j$ would not be adjacent to vertices from $D''$ and hence one can attribute
to it $1$ satisfying $\{0,1\}$-coloring for this graph. What remains to show is that $2\leq|D'|\leq \omega(G)$. To see the first inequality note, that there was $|D'|=1$, then $D'=\{v_0\}$ such that $v_0$ is adjacent to all the vertices of $C_j$. However then $\{v_0\}\cup V_{C_j}$ would form a clique of size $\omega(G)+1$ which contradicts the fact that
the size of maximal clique in $G$ is $\omega(G)$.
Thus by contradiction $|D'|\geq 2$. On the other
hand, since $D'$ was taken to be the minimal
graph dominating $C_j$ with $\omega(G)$ vertices,
it can not be of size more than $\omega(G)$ or 
else it would not be minimal (there are at most
$\omega(G)$ vertices to be dominated in $C_j$). This proves the claim.

\end{proof}

\subsection*{Supplementary Note 4: A generalisation of gadgets with other forbidden value assignments}
To show the precise steps of building a gadget for every forbidden set $\mathcal{H} \subset \{0,1\}^m$ for arbitrary $m \geq 2$
, we first make the following observation to simplify the task at hand.
\begin{lemma}
If for every forbidden set $\mathcal{H} \subset \{0,1\}^m$ of cardinality $|\mathcal{H}| = 1$, and for every set of mutually non-orthogonal vectors $I = \{|v_1 \rangle, \dots, |v_m \rangle \}$, there exists a gadget in dimension $d$ for $\mathcal{H}$ with $I$ as the distinguished set of vectors, then there also exists a gadget in dimension $d$ for forbidden sets $\mathcal{H} \subset \{0,1\}^m$ with $|\mathcal{H}| > 1$.   
\end{lemma}
\begin{proof}
The proof works recursively as follows. Suppose that for every forbidden set $\mathcal{H} \subset \{0,1\}^m$ with $|\mathcal{H}| = 1$ and for every $I = \{|v_1 \rangle, \dots, |v_m \rangle\}$ there exists a gadget $G_{\mathcal{H}, I}$ in dimension $d$ with $I$ as the distinguished set of vectors. To construct a gadget in dimension $d$ for a new forbidden set $\mathcal{H}' \subset \{0,1\}^m$ with $|\mathcal{H}'| = 2$, we simply consider $\mathcal{H}_1 \subset \{0,1\}^m$ and $\mathcal{H}_2 \subset \{0,1\}^m$ such that $|\mathcal{H}_1| = |\mathcal{H}_2| = 1$ and $\mathcal{H}' = \mathcal{H}_1 \cup \mathcal{H}_2$ and the corresponding gadgets $G_{\mathcal{H}_1, I}$ and $G_{\mathcal{H}_2, I}$ with identical distinguished vector sets $I$. It is readily seen that the union of the two graphs $G_{\mathcal{H}_1, I} \cup G_{\mathcal{H}_2, I}$ forms a gadget in dimension $d$ for the forbidden set $\mathcal{H}'$. The construction for any other $\mathcal{H} \subset \{0,1\}^m$ with $|\mathcal{H}| > 1$ follows in an  analogous manner. 
% For a forbidden set $|\mathcal{H}| > 1$ with mutually non-orthogonal vectors $I = \{|v_1 \rangle, \dots, |v_m \rangle \}$.
%Since we know for any each of these forbidden assignments $(f_{I_1},f_{I_2},\cdots,f_{I_{m-1}},f_{I_m})\subset \mathcal{H}$, there exist a forbidden set $|\mathcal{H}_f|=1$ with $I$ as its independent set which satisfy this forbidden assignment $(f_{I_1},f_{I_2},\cdots,f_{I_{m-1}},f_{I_m})$. To make the independent set $I$ satisfy all of these forbidden assignments $\mathcal{H}$, we simply add all of these $|\mathcal{H}_f|=1$ on $I$.
\end{proof}

We now show how to construct gadgets in dimension $d$ for the forbidden sets $\mathcal{H} \subset \{0,1\}^m$ (with $m \leq d$) of cardinality $|\mathcal{H}| = 1$ using the order $(m,m-1)$ gadgets constructed in previous sections. And we illustrate the construction for all $\mathcal{H} \subset \{0,1\}^3$ with $|\mathcal{H}| = 1$ in Fig. \ref{forbidden3}. 

%Firstly, we will show how to use the high order gadget $G_{(k,k-1)}$ ($k\leq\omega(G)$) to construct the general gadget whose number of assignment elements is $|\mathcal{H}|=2^{k}-1$. These are the 1-assignment $(f_{I_1},f_{I_2},\cdots,f_{I_{k-1}},f_{I_k})$ forbidden set, which can also be written as $(f_{I_1},f_{I_2},\cdots,f_{I_{k-1}})\Rightarrow \overline{f_{I_k}}$ set, where $f_{I_i}\in\{0,1\},~ i=1,\cdots,k$ and $\overline{0}=1, \overline{1}=0$. Totally, there are $2^k$ such kind of general gadgets $G$ ($\mathcal{H}$).
\begin{itemize}
    \item First remark that the order $(m,m-1)$ gadgets on their own directly correspond to the all-ones forbidden set $\mathcal{H} = \{(\underbrace{1,\dots, 1}_m)\}$.
    
%    High order gadget $G_{(k,k-1)}$ ($k\leq\omega(G)$) directly corresponds to $\mathcal{H}=\{0,1\}^{\times k}\setminus\{(1,1,\cdots,1,1)\}$, \ie the $(1,1,\cdots,1,1)~forbidden$ set or $ (1,1,\cdots,1)\Rightarrow 0$ set.
    
    \item The construction of the gadget for the forbidden set $\mathcal{H} = \{(\underbrace{1,\dots, 1}_{m-1}, 0)\}$ is illustrated in Fig. \ref{forbidden3}(a). We utilize in this construction an order $(m,m-1)$ gadget from the previous section with vertices $v_1, \dots, v_m$ (illustrated by the dashed triangle in the figure). Now consider a clique of size $d$ with $v_m$ as one of the vertices, and such that each vertex in the clique bar one call it $u$ is adjacent to one of the vertices $v_1, \dots, v_{m-1}$ (such an adjacency structure is possible since the order $(m,m-1)$ gadget construction works for arbitrary vectors $|v_1 \rangle, \dots, |v_m \rangle$ in dimension $d$). It then follows that the combined set of vectors is a gadget for the forbidden set $\mathcal{H}$ with the distinguished vectors being $|v_1 \rangle, \dots, |v_{m-1} \rangle, |u \rangle$.

  %  $\mathcal{H}=\{0,1\}^{\times k}\setminus\{(1,1,\cdots,1,0)\}$, the $(1,1,\cdots,1,0)~forbidden$ set 
    %(note that the position of 0 can be any one of these $k$ vertices) 
    %or $(1,1,\cdots,1)\Rightarrow 1$ set  can be easily construct by a high order gadget $G_{(k,k-1)}$ together with a clique of size $\omega(G)$. One of the vertex $v$ of the high order gadget $G_{(k,k-1)}$ is in this clique, and other $\omega(G)-2$ vertices in the clique are connected by other vertices $v_i,~i=1,\cdots,k-1$ in the high order gadget $G_{(k,k-1)}$. Thus, when all the vertices $v_i,~i=1,\cdots,k-1$ except $v$ in the high order gadget $G_{(k,k-1)}$ are assigned 1, the remaining vertex $u$ in the clique must be 1, \ie it's impossible to assign 0 to this remaining vertex.
    % \begin{figure}[htbp]
    %   \centering
    %   \includegraphics[width=0.45\textwidth]{11_1_n.pdf}\\
    %   \caption{The construction of the $(1,1,\cdots,1)\Rightarrow 1$ set. The dashed circle represents a high order gadget $G_{(k,k-1)}$ and the  solid circle represents a clique of size $\omega(G)$. The red vertices compose the independent set $I$ whose $01$-assignments is $\mathcal{H}=\{0,1\}^{\times k}\setminus\{(1,1,\cdots,1,0)\}$.}
    % \end{figure}

    \item  Any other gadget for a specified forbidden set $\mathcal{H}=\{(f_{I_1},f_{I_2},\ldots,f_{I_{m-1}},f_{I_m})\}$ will be constructed by using the above two gadgets for $\mathcal{H}_1 = \{(\underbrace{1,\dots, 1}_m)\}$ and $\mathcal{H}_2 = \{(\underbrace{1,\dots, 1}_{m-1},0)\}$.

    We begin with $m-1$ bases sets in dimension $d$, denoted as $C_1, C_2, \dots, C_{m-1}$. We choose these sets such that no two vectors in different bases sets are identical or orthogonal to each other (one can do this by picking a single basis set $C_1$ and applying a suitable unitary matrix $U_d$ to $C_1$). We also choose a unit vector $|u\rangle$, that is not orthogonal to any vector in any of the bases $C_i,i\in[m-1]$.
    % We begin with $m-1$ disjoint cliques of size $d$, denoted as $C_1,C_2,\ldots,C_{m-1}$, and one vertex $u$, which is not adjacency to any vertex in clique $C_i,i\in[m-1]$. 
    Select one vector $|v_i\rangle$ in each basis $C_i, i\in[m-1]$, these vectors $|v_i\rangle,i\in[m-1]$ together with $|u \rangle$, will form the mutually non-orthogonal set $I=\{|v_1\rangle,|v_2\rangle,\ldots,|v_{m-1}\rangle,|u\rangle\}$ for the forbidden set $\mathcal{H}=\{(f_{I_1},f_{I_2},\ldots,f_{I_{m-1}},f_{I_m})\}$.
    
    Now construct all possible sets $T_j$ consisting of elements $t_{ji},i\in[m-1]$ from each of the bases $C_i$ and such that $t_{ji}=|v_i\rangle$ if and only if $f_{|v_i\rangle}=1$ in the forbidden assignment $\mathcal{H}$. Thus, we have $(d-1)^{m-1-\sum_{i=1}^{m-1} f_{|v_i\rangle}}$ sets $T_j$ with $|T_j|=m-1$ for each $j\in[(d-1)^{m-1-\sum_{i=1}^{m-1} f_{|v_i\rangle}}]$. As the last step of the construction, for each $j\in[(d-1)^{m-1-\sum_{i=1}^{m-1} f_{|v_i\rangle}}]$, construct a gadget for the forbidden set $H'=\{(\underbrace{1,\dots, 1}_{m-1},f_{|u\rangle})\}$ with the vectors in $T_j\cup\{|u\rangle\}$ being the distinguished vectors (where $f_{|u\rangle}$ is the forbidden assignment for the vector $|u\rangle$). Such gadgets for the forbidden set $H'=\{(\underbrace{1,\dots, 1}_{m-1},f_{|u\rangle})\}$ and distinguished vectors  $T_j\cup\{|u\rangle\}$ are illustrated with the dashed triangles in Fig. \ref{forbidden3}(b) and (c).

%  And in the forbidden assignment bit string $(f_{I_1},f_{I_2},\cdots,f_{I_{k-1}})$, for all $i=1,\cdots,k-1$ if $f_{v_i}=1$, then add $v_i$ into the set $Temp$. Otherwise, if $f_{v_i}=0$, choose any one vertex $v_{ij}$ in clique $C_i$ (but not $v_i$) and add it into the set $Temp$. Suppose there are $n$ numbers of $0$ in the forbidden assignment bit string $(f_{I_1},f_{I_2},\cdots,f_{I_{k-1}})$, then totally there are $(\omega(G)-1)^n$ possible temporary sets $Temp$. The last step of construction is that for any possible set $Temp$,  all the vertices in this set $Temp$ and the special vertex $u$ are independent and connected by a $ (1,1,\cdots,1)\Rightarrow \overline{f_{I_k}}$ (note that $u$ is always the last vertex). If these vertices are not independent, then put them aside and do nothing.
\end{itemize}  
Remark that the above gadget constructions for forbidden sets $\mathcal{H}$ with  $|\mathcal{H}|=1$ are valid for any arbitrary set of mutually non-orthogonal vectors $I = \{|v_1 \rangle, \dots, |v_m \rangle \}$ in dimension $d$, since all these constructions use the order $(m,m-1)$ gadgets from the previous section which work for arbitrary mutually non-orthogonal vectors. Finally, it's worth noting that Kochen Specker proofs themselves come under the umbrella of the generalised gadget structures defined here, with $\mathcal{H} = \{0,1\}^m$ for some independent set in the graph $I$ of size $|I| = m$ and for arbitrary $m \geq 1$. In other words, there is no valid $\{0,1\}$-assignment to the vertices of the independent set $I$.

%and the $(m,m-1)$ gadget construction works for arbitrary mutually non-orthogonal vectors in dimension $d$.
    
% We illustrate the above construction for all $\mathcal{H} \subset \{0,1\}^3$ with $|\mathcal{H}| = 1$ in Fig. \ref{forbidden3}. 

\begin{figure}[t]
\centering
\subfigure[$\mathcal{H}=\{(1,1,0)\}$]{
\begin{minipage}[t]{0.3\textwidth}
\centering
\includegraphics[width=0.9\textwidth]{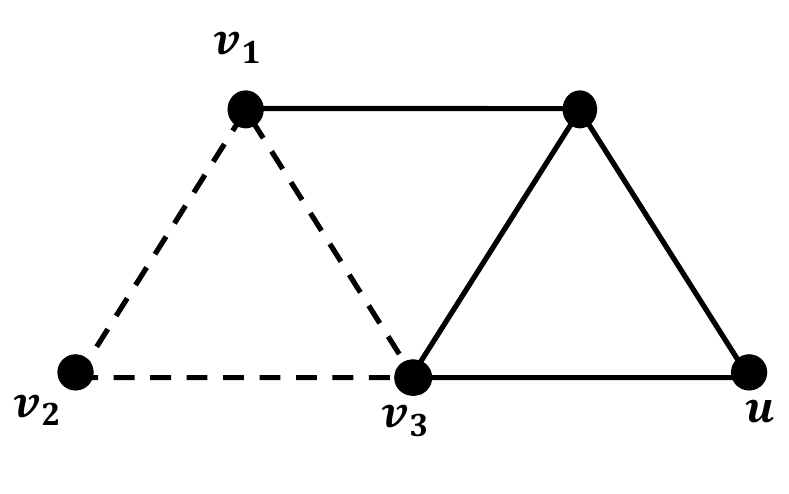}
\end{minipage}
}
\subfigure[$\mathcal{H}=\{(1,0,f_u)\}$ or $\mathcal{H}=\{(0,1,f_u)\}$
(swap the position of $v_1,v_2$) ]{
\begin{minipage}[t]{0.3\textwidth}
\centering
\includegraphics[width=0.9\textwidth]{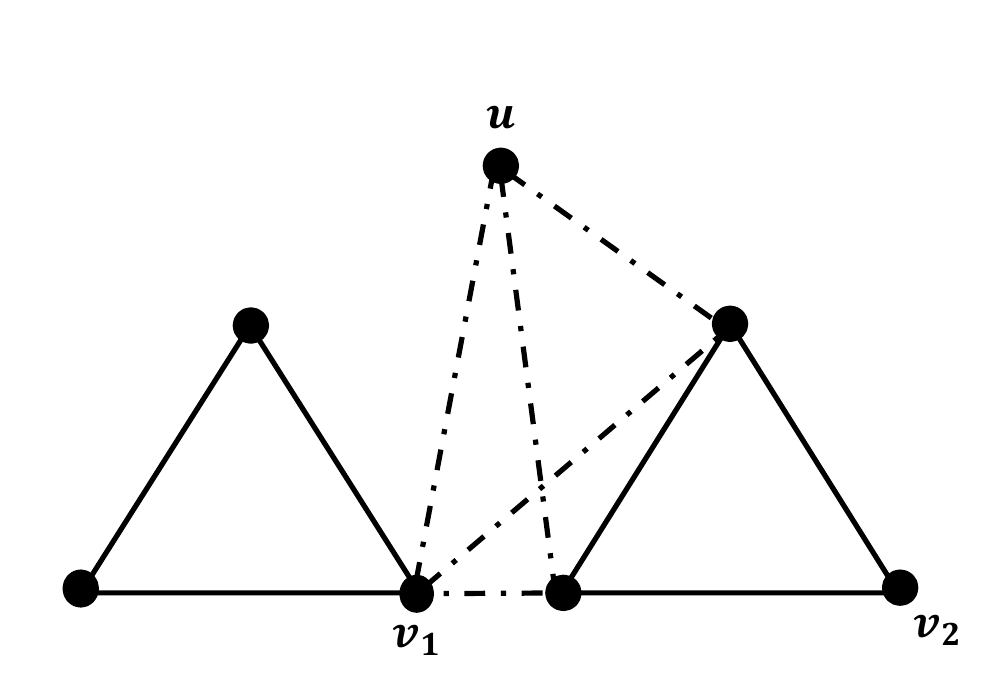}
\end{minipage}
}
\subfigure[$\mathcal{H}=\{(0,0,f_u)\}$]{
\begin{minipage}[t]{0.3\textwidth}
\centering
\includegraphics[width=0.9\textwidth]{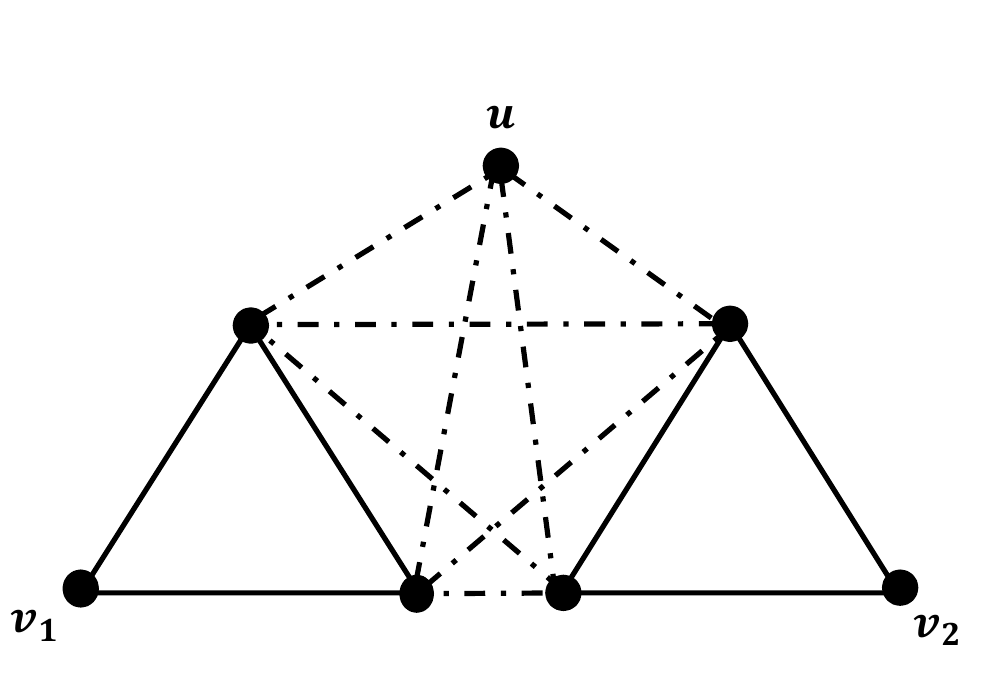}
\end{minipage}
}
\caption{The constructions of gadgets for forbidden sets $\mathcal{H} \subset \{0,1\}^3$ with $|\mathcal{H}| = 1$. The dashed triangle in (a) represents the order $(3,2)$ gadget constructed in the previous section, the assignment $(1,1,0)$ if forbidden for the independent set of vertices $\{v_1, v_2, u\}$. Each dash-dotted triangle in (b) and (c) corresponds to a gadget for the forbidden set $\mathcal{H}=\{(1,1,f_u)\}$ constructed previously in this section. Again the assignment in $\mathcal{H}$ is forbidden for the independent set of vertices $\{v_1, v_2, u\}$. }
\label{forbidden3}
\end{figure}

\end{document}